\documentclass[12pt,english]{article}
\usepackage[T1]{fontenc}
\usepackage[latin9]{inputenc}
\usepackage{amsmath}
\usepackage{amsthm}
\usepackage{amssymb}
\usepackage{geometry}
\usepackage{setspace}
\usepackage[authoryear]{natbib}
\makeatletter
\@ifundefined{date}{}{\date{}}
\usepackage{threeparttable} 
\usepackage{float}
\providecommand{\U}[1]{\protect\rule{.1in}{.1in}}

\renewenvironment{thebibliography}[1]{
  \begin{oldthebibliography}{#1}
    \setlength{\itemsep}{0em}
    \setlength{\parskip}{0em}
}
{
  \end{oldthebibliography}
}

\makeatother

\theoremstyle{remark}
\newtheorem{rem}{\protect\remarkname}
\theoremstyle{plain}
\newtheorem{prop}{\protect\propositionname}
\newtheorem{assumption}{\protect\assumptionname}
\newtheorem{thm}{\protect\theoremname}
\newtheorem{lem}{\protect\lemmaname}
\newtheorem{cor}{\protect\corollaryname}
\usepackage{babel}
\providecommand{\assumptionname}{Assumption}
\providecommand{\corollaryname}{Corollary}
\providecommand{\lemmaname}{Lemma}
\providecommand{\propositionname}{Proposition}
\providecommand{\remarkname}{Remark}
\providecommand{\theoremname}{Theorem}

\begin{document}
\title{Estimation and Inference on Average Treatment Effect in Percentage Points under Heterogeneity\thanks{The estimation and inference methods discussed in this paper are implemented in an accompanying R package and Stata command. The Stata package \texttt{ate\_pct} is available from the SSC repository, with the latest version available at https://github.com/zengying17/ate\_pct-stata. The R package \texttt{atepct} is available at https://github.com/zengying17/ate\_pct-r. Illustrative replication code can be found at https://github.com/zengying17/ate\_pct-replication.}}
\author{Ying Zeng\thanks{E-mail address: zengying17@xmu.edu.cn.}\\
School of Economics and Wang Yanan Institute for Studies in Economics, \\
Xiamen University, China}
\maketitle
\begin{abstract}
In semi-logarithmic regressions, treatment coefficients are often interpreted as approximations of an average treatment effect (ATE) in percentage points. This paper highlights the overlooked bias of this approximation under treatment effect heterogeneity, arising from Jensen's inequality. The issue is particularly relevant for difference-in-differences designs with log-transformed outcomes and staggered treatment adoption, where treatment effects may vary across groups and periods. This paper proposes new estimation and inference methods for an estimand that accounts for heterogeneity across observable subgroups and can improve upon conventional measures. The estimand provides a lower bound on the ATE in percentage points for the relevant target (sub)population, and coincides with it in the absence of within-group heterogeneity. I establish the methods' large-sample properties and study their finite-sample performance through Monte Carlo experiments, which reveal substantial discrepancies between conventional and proposed measures when systematic heterogeneity is large. Two empirical applications further underscore the practical importance of these methods.

\end{abstract}
\global\long\def\diag{\operatorname{diag}}%
\global\long\def\tr{\operatorname{tr}}%
\global\long\def\Var{\operatorname{Var}}%
\global\long\def\det{\operatorname{det}}%
\global\long\def\Cov{\operatorname{Cov}}%
\global\long\def\argmin{\operatorname{argmin}}%
\global\long\def\plim{\operatorname{plim}}%
\global\long\def\exp{\operatorname{exp}}%
\global\long\def\hyp{\operatorname{_{0}F_{1}}}%

\section{Introduction}

Semi-logarithmic regression models, where the dependent variable is in the natural logarithm form, are widely used in empirical studies. In these models, a small coefficient on a binary treatment is often interpreted as an approximate average treatment effect (ATE) in percentage terms, possibly for some subpopulation such as treated units or compliers for an instrument \citep{chen_logs_2024,hansen_econometrics_2022}. This interpretation is based on the following logic: Let $Y_{1i}$ and $Y_{0i}$ denote the potential outcomes under treatment and non-treatment respectively for individual $i$. The treatment effect in log points is defined as $\tau_{i}=\ln(Y_{1i})-\ln(Y_{0i})$, while the percentage change is given by $\rho_{i}=\left(Y_{1i}-Y_{0i}\right)/Y_{0i}=\exp(\tau_{i})-1$. When $|\tau_{i}|$ is small, $\exp(\tau_{i})-1\approx\tau_{i}$, justifying the interpretation of $\tau_{i}$ as an approximate percentage effect.

When treatment effects are heterogeneous, however, neither the ATE in log points, $\bar{\tau}:=E(\tau_{i})$, nor $\exp(\bar{\tau})-1$ can be interpreted as the ATE in percentage points $\bar{\rho}:=E(\rho_{i})$. The intuition is straightforward: Jensen's inequality implies $E[\exp(\tau_{i})]-1\geqslant\exp[E(\tau_{i})]-1$, hence $\bar{\rho}\geqslant\exp(\bar{\tau})-1$, with equality holding if and only if log-point treatment effects are constant. Consequently, $\exp(\bar{\tau})-1<\bar{\rho}$ under treatment effect heterogeneity, with disparities increasing as treatment effects exhibit greater variation.

The bias that arises when the ATE in log points $\bar{\tau}$ or $\exp(\bar{\tau})-1$ is interpreted as the ATE in percentage points has received limited attention in the literature, except for \citet{chen_logs_2024} who show that $\bar{\rho}$ cannot be point identified if arbitrary heterogeneity is allowed. This oversight may lead to misleading interpretations of results in empirical studies with heterogeneous treatment effects and log-transformed outcomes. A prominent example is the semi-log difference-in-differences model with staggered treatment adoption, where different groups start receiving treatment at different times. Recent work has proposed heterogeneity-robust estimators for such models, whose estimands can be expressed as convex weighted averages of group-time treatment effects (\citealp{borusyak_revisiting_2024,callaway_differenceindifferences_2021,dechaisemartin_twoway_2020,sun_estimating_2021}, etc.). When outcome variables are log-transformed, these estimators still represent some form of the ATE in log points $\bar{\tau}$. Consequently, these estimators and their exponential minus one should not be interpreted as the ATE in percentage points.

This paper highlights the often-overlooked problem of interpreting the ATE in log points as the ATE in percentage points in the context of heterogeneous treatment effects. Unfortunately, $\bar{\rho}$ cannot be point identified without additional distributional assumptions because $E\left[\left(Y_{1}-Y_{0}\right)/Y_{0}\right]$ is a functional of the joint distribution of $(Y_{1},Y_{0})$. Since $Y_{1}$ and $Y_{0}$ are not simultaneously observed for the same unit, only their marginal distributions are identified, and such functionals generally cannot be point identified \citep{callaway_bounds_2021,fan_partial_2017,heckman_making_1997}. In particular, Proposition 3 of \citet{chen_logs_2024} provides a formal proof that $\bar{\rho}$ is not point identified when heterogeneity is unrestricted. A growing literature develops partial identification approaches for such functionals, including \citet{callaway_bounds_2021}, \citet{fan_partial_2017}, \citet{firpo_partial_2019}, and \citet{frandsen_partial_2021}. Building on the framework of \citet{fan_partial_2017}, this paper derives sharp bounds for $\bar{\rho}$ under the assumption that $(Y_{1},Y_{0})$ is jointly independent of treatment conditional on observables. However, estimating these sharp bounds requires the relevant marginal distributions of the potential outcomes to be identified. This requirement is not generally satisfied in difference-in-differences designs under standard assumptions. Moreover, even in settings where the required marginals are identified, such as randomized experiments, estimation can be statistically challenging and involves choices such as the kernel and bandwidth. Point identification of $\bar{\rho}$ under heterogeneity is possible only under additional distributional assumptions. For example, the Online Appendix discusses estimation and inference of $\bar{\rho}$ under joint normality of the potential outcomes, an assumption that can be restrictive in practice.

Given the inherent difficulty of point identifying $\bar{\rho}$ under unrestricted heterogeneity, this paper proposes estimation and inference methods for $\rho_{b}=\sum_{g}w_{g}\left[\exp(\tau_{g})-1\right]$, a subgroup-weighted average of transformed average log-point treatment effects, where $w_{g}$ and $\tau_{g}$ are, respectively, the share and average log-point effect of subgroup $g$. By accounting for heterogeneity across observable subgroups, this estimand can improve upon the conventional measures $\exp(\bar{\tau})-1$ and $\bar{\tau}$ and remains point identified. When treatment effects are constant within subgroups, $\rho_{b}$ equals $\bar{\rho}$. When treatment effects are heterogeneous within subgroups, $\rho_{b}$ remains a valid lower bound for $\bar{\rho}$ that is at least as tight as $\exp(\bar{\tau})-1$ and $\bar{\tau}$, and is strictly tighter when subgroup average log-point effects vary.

The proposed estimation and inference methods for $\rho_{b}$ do not require strong distributional assumptions on the potential outcomes and are straightforward to implement. They are developed under a general econometric framework that encompasses semi-log regression models and semi-log staggered difference-in-differences designs. Their large-sample validity is established under standard regularity conditions, and their finite-sample performance is demonstrated through Monte Carlo simulations. To illustrate these methods' practical relevance and applicability, the paper presents two empirical applications: one uses a semi-log regression model to study the causal impact of exporting on firm productivity, and the other uses a staggered difference-in-differences design to examine the effect of water and sewerage systems on child mortality. Through these contributions, this paper provides researchers with tools to more accurately estimate and interpret treatment effects in percentage terms, particularly in the presence of heterogeneity.

In addition to the literature related to partial identification of $\bar{\rho}$, this paper also connects to two other strands of the literature. First, it contributes to a growing body of work on treatment effect heterogeneity, which has shown that conventional estimates assuming a constant treatment effect are weighted averages of treatment effects, with weights possibly negative and not equal to the true population share. Examples include ordinary least squares (OLS) estimators \citep{angrist_estimating_1998,gibbons_broken_2019,goldsmith-pinkham_contamination_2024,sloczynski_interpreting_2022}, two-stage least squares (2SLS) estimators \citep{imbens_identification_1994,mogstad_causal_2021}, and two-way fixed effects estimators for staggered difference-in-differences models \citep{borusyak_revisiting_2024,dechaisemartin_twoway_2020,goodman-bacon_differenceindifferences_2021,sun_estimating_2021}. Numerous heterogeneity-robust estimators have been developed, typically by estimating treatment effects and population shares for subgroups, and then aggregating them up to obtain an ATE for a target population (e.g., \citealp{callaway_differenceindifferences_2021,gibbons_broken_2019,goldsmith-pinkham_contamination_2024,sun_estimating_2021}). The results presented in this paper imply that when the outcome is log-transformed, the estimands of these approaches do not generally correspond to the average proportional treatment effect under heterogeneity, but instead yield non-sharp lower bounds for it. However, the heterogeneity-robust estimators in these papers form the basis for the estimation and inference methods proposed in the paper.

Second, this work supplements ongoing research into the potential pitfalls of using logarithmic transformations in empirical analyses. For instance, \citet{chen_logs_2024} and \citet{mullahy_why_2024} demonstrated that when dependent variables have many zero values, coefficients in regressions with log-like transformations such as $\ln(Y+1)$ do not have the interpretation of percentage effects. \citet{manning_logged_1998} and \citet{silva_log_2006} highlighted issues of using log-linear regressions for estimating the impact on the outcome's mean and elasticity under heteroscedasticity. The results in this paper further demonstrate that in the presence of treatment effect heterogeneity, coefficients in semi-log regressions may not have an ATE in percentage points interpretation.

The remainder of the paper is structured as follows. Section \ref{sec:ATE-in-Percentage} develops the framework for analyzing ATE in percentage points under treatment effect heterogeneity. It first characterizes the relationship between log-point and percentage-point effects and the bias induced by Jensen's inequality. It then proposes estimation and inference methods for a group-aggregated estimand, $\rho_{b}$, which accommodates heterogeneity across observable subgroups. Finally, it derives sharp bounds for the average proportional effect under partial identification. Section \ref{sec:MC} reports the results of Monte Carlo simulations for the proposed methods, while Section \ref{sec:Empirical-Applications} presents two empirical applications. Section \ref{sec:Conclusion} concludes. The Online Appendix provides detailed discussions of examples of the proposed methods, discusses estimation and inference of $\bar{\rho}$ under joint normality of potential outcomes, introduces a simple bias-correction method, and presents additional Monte Carlo and empirical results.

\section{ATE in Percentage Points: Framework and Identification\label{sec:ATE-in-Percentage}}

\subsection{Log-Point and Percentage-Point Treatment Effects under Heterogeneity\label{subsec:set_up}}

This section establishes the formal relationship between the ATE in log points and the ATE in percentage points. The setup allows log-point treatment effects to be heterogeneous both across and within groups, commonly referred to in the literature as systematic and idiosyncratic heterogeneity, respectively \citep{djebbari_heterogeneous_2008}.

Suppose we are interested in the average treatment effect in percentage points for a target (sub)population $P$. The target (sub)population may be the full population, the treated subpopulation, or another policy-relevant population. Let $P$ be partitioned into $G$ mutually exclusive and collectively exhaustive subgroups, indexed by $g=1,\ldots,G$, where $G$ is finite. Let $w_{g}$ denote the share of subgroup $g$ in $P$, where $w_{g}>0$ and $\sum^{G}_{g=1}w_{g}=1$. Let $D^{(g)}_{i}=1$ if individual $i$ belongs to subgroup $g$ of the target (sub)population, and 0 otherwise. Let $Y_{0i}$ and $Y_{1i}$ represent the potential outcomes for individual $i$ if not treated and if treated respectively. I assume that $Y_{0i}>0$ and $Y_{1i}>0$ for all $i$, so that their natural logarithms are well defined.
\begin{rem}
When outcome variables contain many zero values, alternative frameworks may be more appropriate. Specifically, Poisson regression models, as proposed by \citet{chen_logs_2024} and \citet{mullahy_why_2024}, can be used to estimate the ATE in levels as a percentage of the baseline mean.
\end{rem}
The individual treatment effect in log points is defined as $\tau_{i}=\ln\left(Y_{1i}\right)-\ln\left(Y_{0i}\right)$, and the average log-point effect for group $g$ is $\tau_{g}=E_{P}(\tau_{i}|D^{(g)}_{i}=1)$, assuming this conditional expectation exists and is finite. The ATE in log points for the target (sub)population $P$ is
\begin{equation}
\bar{\tau}=E_{P}(\tau_{i})=\sum^{G}_{g=1}w_{g}\tau_{g}.\label{eq:taubar}
\end{equation}
 The individual treatment effect in percentage points is $\rho_{i}=\left(Y_{1i}-Y_{0i}\right)/Y_{0i}=\exp(\tau_{i})-1$, and the ATE in percentage points for group $g$ is $\rho_{g}=E_{P}(\rho_{i}|D^{(g)}_{i}=1)=E_{P}(\exp(\tau_{i})|D^{(g)}_{i}=1)-1$. The ATE in percentage points for the target (sub)population $P$ is 
\begin{equation}
\bar{\rho}=E_{P}(\rho_{i})=\sum^{G}_{g=1}w_{g}\rho_{g}.\label{eq:rhobar}
\end{equation}

The identity $\rho_{i}=\exp(\tau_{i})-1$ motivates the common practice of reporting
\[
\rho_{a}=\exp(\bar{\tau})-1
\]
as the ATE in percentage points. Alternatively, when $\bar{\tau}$ is close to zero, $\bar{\tau}$ itself is often used as an approximation to $\bar{\rho}$, since $\exp(x)-1\approx x$ for small $x$. More generally, it holds that $\exp(\bar{\tau})-1\geqslant\bar{\tau}$, with equality holding if and only if $\bar{\tau}=0$, because the function $\exp(x)-x-1$ is strictly convex and attains its unique minimum of zero at $x=0$.

However, Jensen's inequality implies $E_{P}[\exp(\tau_{i})]\geqslant\exp[E_{P}(\tau_{i})]$ and hence $\bar{\rho}\geqslant\rho_{a}$, with equality holding if and only if $\tau_{i}=\bar{\tau}$ a.s. Under treatment effect heterogeneity, $\bar{\rho}>\rho_{a}$. The gap can arise from within-group heterogeneity or from cross-group heterogeneity. First, due to within-group heterogeneity, for $g=1,\ldots,G$, we have $E_{P}\left[\exp(\tau_{i})|D^{(g)}_{i}=1\right]\geqslant\exp\left[E_{P}\left(\tau_{i}|D^{(g)}_{i}=1\right)\right]=\exp(\tau_{g})$, which implies that $\rho_{g}\geqslant\exp(\tau_{g})-1$. Define
\begin{equation}
\rho_{b}=\sum^{G}_{g=1}w_{g}\exp(\tau_{g})-1.\label{eq:rho_b}
\end{equation}
Then $\bar{\rho}\geqslant\rho_{b}$, with equality holding if and only if $\tau_{i}$ is constant within each subgroup. Second, due to heterogeneity across groups, $\sum^{G}_{g=1}w_{g}\exp(\tau_{g})\geqslant\exp\left(\sum^{G}_{g=1}w_{g}\tau_{g}\right)$, which implies $\rho_{b}\geqslant\exp(\bar{\tau})-1$, with equality holding if and only if $\tau_{g}=\bar{\tau}$ for all $g$. These results are summarized in the following proposition.
\begin{prop}
\label{prop:rhoplus}The ATE in percentage points satisfies 
\[
\bar{\rho}\geqslant\rho_{b}\geqslant\exp(\bar{\tau})-1\geqslant\bar{\tau},
\]
where the first inequality holds with equality if and only if $\Var(\tau_{i}|D^{(g)}_{i}=1)=0$ for all $g$; the second holds with equality if and only if $\tau_{g}=\bar{\tau}$ for all $g$; and the third holds with equality if and only if $\bar{\tau}=0$.
\end{prop}
Thus, under either systematic or idiosyncratic heterogeneity, neither the average log-point effect $\bar{\tau}$ nor $\exp(\bar{\tau})-1$ equals the average proportional effect $\bar{\rho}$, although $\exp(\bar{\tau})-1$ is generally less biased than $\bar{\tau}$.

Identifying $\bar{\rho}$ is challenging because it is a functional of the joint distribution of potential outcomes; see the Introduction for a discussion of the related literature. Using a partial identification approach developed by \citet{fan_partial_2017}, Section \ref{subsec:Partial-Identification} derives sharp upper and lower bounds for $\bar{\rho}$. However, this partial identification approach requires the relevant marginal distributions, or conditional marginal distributions when covariates are used, of $Y_{0i}$ and $Y_{1i}$ to be identified. This requirement is satisfied in randomized experiments, but is not generally satisfied in DID settings under standard assumptions. Even if these distributions are identified, estimation can be statistically demanding and requires choices such as kernels and bandwidths. In Section \ref{sec:normal heteogeneity} of the Online Appendix, I also discuss the estimation and inference of $\bar{\rho}$ under the assumption that $(Y_{0},Y_{1})$ is jointly normal. Although this approach delivers point identification of $\bar{\rho}$ under both systematic and idiosyncratic heterogeneity, the assumption of joint normality of potential outcomes can be too restrictive in practice and limits its applicability.

Because idiosyncratic heterogeneity is inherently difficult to characterize, this paper focuses on $\rho_{b}$, which captures heterogeneity across groups. When treatment effects are constant within each subgroup, the average proportional effect $\bar{\rho}$ is point identified and coincides with $\rho_{b}$. When within-group heterogeneity is present, Proposition \ref{prop:rhoplus} implies that $\rho_{b}$ remains a valid lower bound for $\bar{\rho}$. This lower bound is at least as sharp as $\exp(\bar{\tau})-1$ or $\bar{\tau}$, and is sharper under cross-group heterogeneity. The difference between $\rho_{b}$ and $\bar{\rho}$ therefore reflects residual idiosyncratic heterogeneity.

Accordingly, this paper develops estimation and inference methods for $\rho_{b}$, as discussed in Section \ref{sec:Estimation-and-Inference}. Compared to partial identification approaches and methods relying on joint normality assumptions, the proposed methods are straightforward to implement and rely on standard assumptions, providing a practical approach to analyzing treatment effects in percentage points under heterogeneity.

\subsection{Estimation and Inference on $\rho_{b}$\label{sec:Estimation-and-Inference}}

This section proposes a general econometric framework for the estimation and inference of $\rho_{b}$ in Eq. (\ref{eq:rho_b}), building upon consistent and asymptotically normal estimators of $w=(w_{1},\ldots,w_{G})^{\prime}$ and $\tau=(\tau_{1},\ldots,\tau_{G})^{\prime}$. The proposed methodology imposes no restrictions on the functional form or application context of these estimators, provided they satisfy the asymptotic condition stated below.

I maintain the setup introduced in Section \ref{subsec:set_up}. Specifically, the number of groups $G$ is finite, groups with zero weight in the target population are excluded from analysis, the weights $w_{g}$ are positive and sum to one, and each subgroup average log-point effect $\tau_{g}$ is well-defined and finite. These conditions ensure that $\bar{\tau}$ and $\rho_{b}$ are well defined and bounded. I allow $w_{g}=1$ to incorporate the case of $G=1$, representing no cross-group heterogeneity.

Let $\hat{w}=(\hat{w}_{1},\ldots,\hat{w}_{G})^{\prime}$ and $\hat{\tau}=(\hat{\tau}_{1},\ldots,\hat{\tau}_{G})^{\prime}$ denote estimators of $w$ and $\tau$ respectively.
\begin{assumption}
\label{assu:delta}Let $\delta=(w^{\prime},\tau^{\prime})^{\prime}$ and $\hat{\delta}=(\hat{w}^{\prime},\hat{\tau}^{\prime})^{\prime}$. There exists a $2G\times2G$ finite matrix $\bar{\Sigma}_{\delta}$ and its estimator $\hat{\bar{\Sigma}}_{\delta}$ such that 
\[
\sqrt{N}\left(\hat{\delta}-\delta\right)\xrightarrow{d}\mathcal{N}\left(0_{2G\times1},\bar{\Sigma}_{\delta}\right),
\]
and $\hat{\bar{\Sigma}}_{\delta}\xrightarrow{p}\bar{\Sigma}_{\delta}$ as the sample size $N$ goes to infinity.
\end{assumption}
Assumption \ref{assu:delta} accommodates both known and estimated weights. In the special case where $w$ is known, we have $\hat{w}=w$. The asymptotic variance of $\sqrt{N}(\hat{w}-w)$ is therefore $\bar{\Sigma}_{w}=0$ and $\Cov(\hat{w},\hat{\tau})=0$. This scenario arises naturally when subgroup weights are fixed by the target population or predetermined by design, such as staggered difference-in-differences designs where treatment effects for a given cohort are averaged across periods with equal and predetermined weights. Under known weights, $\bar{\Sigma}_{\delta}=\diag\{0_{G\times G},\bar{\Sigma}_{\tau}\}$ and $\hat{\bar{\Sigma}}_{\delta}=\diag\{0_{G\times G},\hat{\bar{\Sigma}}_{\tau}\}$, where $\bar{\Sigma}_{\tau}$ is the asymptotic variance of $\sqrt{N}(\hat{\tau}-\tau)$ and $\hat{\bar{\Sigma}}_{\tau}$ is its estimator. Assumption \ref{assu:delta} then reduces to $\sqrt{N}(\hat{\tau}-\tau)\xrightarrow{d}\mathcal{N}(0_{G\times1},\bar{\Sigma}_{\tau})$ and $\hat{\bar{\Sigma}}_{\tau}\xrightarrow{p}\bar{\Sigma}_{\tau}$. In contrast, when working with a random sample from a population, the estimated group weights are subject to sampling variability so that $\bar{\Sigma}_{w}\neq0$.

Given the consistency of $\hat{w}$ and $\hat{\tau}$ under Assumption \ref{assu:delta}, a natural estimator for $\rho_{b}=\sum^{G}_{g=1}w_{g}\exp(\tau_{g})-1$ is: 
\begin{equation}
\hat{\rho}_{b}=\sum^{G}_{g=1}\hat{w}_{g}\exp(\hat{\tau}_{g})-1.\label{eq:rho_b_hat}
\end{equation}
By the continuous mapping theorem, $\hat{\rho}_{b}\xrightarrow{p}\rho_{b}$, where $\rho_{b}$ is a lower bound for $\bar{\rho}$ according to Proposition \ref{prop:rhoplus}.

To conduct inference, I derive the asymptotic distribution of $\hat{\rho}_{b}$ using the delta method. Note that $\partial\rho_{b}/\partial w_{g}=\exp(\tau_{g})$ and $\partial\rho_{b}/\partial\tau_{g}=w_{g}\exp(\tau_{g}),$ so the gradient of $\rho_{b}$ with respect to $\delta=(w^{\prime},\tau^{\prime})^{\prime}$ is
\[
\nabla_{\delta}\rho_{b}=\left[\begin{array}{l}
\exp(\tau)\\
w\odot\exp(\tau)
\end{array}\right],
\]
where $\exp(\tau)=(\exp(\tau_{1}),\ldots,\exp(\tau_{G}))^{\prime}$, and $\odot$ denotes the element-wise product. With $\sqrt{N}\left(\hat{\delta}-\delta\right)\xrightarrow{d}\mathcal{N}(0,\bar{\Sigma}_{\delta})$, the delta method yields $\sqrt{N}(\hat{\rho}_{b}-\rho_{b})\xrightarrow{d}\mathcal{N}(0,\bar{\sigma}^{2}_{b})$, where
\begin{equation}
\bar{\sigma}^{2}_{b}=\left(\nabla_{\delta}\rho_{b}\right)^{\prime}\bar{\Sigma}_{\delta}\left(\nabla_{\delta}\rho_{b}\right)=\left[\begin{array}{l}
\exp(\tau)\\
w\odot\exp(\tau)
\end{array}\right]^{\prime}\bar{\Sigma}_{\delta}\left[\begin{array}{l}
\exp(\tau)\\
w\odot\exp(\tau)
\end{array}\right].\label{eq:Sigma_b}
\end{equation}
By the continuous mapping theorem, a consistent estimator of $\bar{\sigma}^{2}_{b}$ is 
\begin{equation}
\hat{\bar{\sigma}}^{2}_{b}=\left[\begin{array}{l}
\exp(\hat{\tau})\\
\hat{w}\odot\exp(\hat{\tau})
\end{array}\right]^{\prime}\hat{\bar{\Sigma}}_{\delta}\left[\begin{array}{l}
\exp(\hat{\tau})\\
\hat{w}\odot\exp(\hat{\tau})
\end{array}\right].\label{eq:hat_var_rhob}
\end{equation}
When weights are known, $\bar{\Sigma}_{\delta}=\diag\{0_{G\times G},\bar{\Sigma}_{\tau}\}$ and the asymptotic variance simplifies to $\bar{\sigma}^{2}_{b}=\left(w\odot\exp(\tau)\right)^{\prime}\bar{\Sigma}_{\tau}\left(w\odot\exp(\tau)\right)$.
\begin{thm}
\label{thm:rhob} Suppose Assumption \ref{assu:delta} holds, then $\sqrt{N}(\hat{\rho}_{b}-\rho_{b})\xrightarrow{d}\mathcal{N}\left(0,\bar{\sigma}^{2}_{b}\right)$, and $\hat{\bar{\sigma}}^{2}_{b}\xrightarrow{p}\bar{\sigma}^{2}_{b}$.
\end{thm}
The proof follows directly from the previous discussions. This theorem establishes the basis for hypothesis testing and confidence interval construction. To test $H_{0}:\rho_{b}=\rho_{0}$, we can use the $z$-statistic $z_{\rho}=\left(\hat{\rho}_{b}-\rho_{0}\right)/\sigma_{b}$, where $\sigma_{b}=\bar{\sigma}_{b}/\sqrt{N}$ denotes the asymptotic standard error of $\hat{\rho}_{b}$ and is usually estimated by $\hat{\bar{\sigma}}_{b}/\sqrt{N}$ in practice.

Theorem \ref{thm:rhob} implies that $\hat{\rho}_{b}$ consistently estimates $\rho_{b}$, which is a lower bound for the target ATE in percentage points. Thus $\hat{\rho}_{b}$ has an asymptotic lower-bound interpretation. This interpretation requires enough observations in each subgroup for $\hat{w}$ and $\hat{\tau}$ to be consistent. Valid inference further requires $(\hat{w}^{\prime},\hat{\tau}^{\prime})^{\prime}$ to be well approximated by its joint limiting normal distribution. With more subgroups, these requirements become increasingly difficult to meet. Analysts therefore should avoid defining overly disaggregated subgroups. In staggered difference-in-differences designs, adoption cohorts may be inappropriate subgroups when they are small. Monte Carlo simulations in Section \ref{sec:MC} show that the proposed estimation and inference procedures perform well with modest sample sizes. For applications with particularly small group sizes, the bias-corrected estimator in Section \ref{sec:rho_c} of the Online Appendix offers additional improvements.
\begin{rem}
\label{rem:imputation}An extreme case of over-disaggregation is to define each observation as its own subgroup. In this case, $\rho_{b}=\bar{\rho}$ because there is no within-subgroup heterogeneity. However, $\hat{\rho}_{b}$ is inconsistent for $\rho_{b}=\bar{\rho}$ because each $\tau_{g}$ is not consistently estimated. For example, in staggered difference-in-differences designs with log-transformed outcomes, imputation estimators such as \citet{borusyak_revisiting_2024} can produce individual-time specific log-point treatment effect estimates $\hat{\tau}_{it}$ for treated units. Under a strong parallel trends assumption, $\hat{\tau}_{it}$ can be unbiased for $\tau_{it}$ but remains inconsistent. Thus the corresponding $\hat{\rho}^{imp}_{b}$ constructed from averaging $\exp(\hat{\tau}_{it})-1$ over treated cells is inconsistent. Moreover, Jensen's inequality implies that $E[\exp(\hat{\tau}_{it})]\geqslant\exp[E(\hat{\tau}_{it})]=\exp(\tau_{it})$ because of estimation error. Thus $E(\hat{\rho}^{imp}_{b})\geqslant\rho_{b}=\bar{\rho}$, so $E(\hat{\rho}^{imp}_{b})$ is an upper bound for $\bar{\rho}$ in this case.
\end{rem}
The estimation and inference methods developed in this section apply broadly, as Assumption \ref{assu:delta} holds in many commonly used empirical settings. One example is the semi-log regression model. Suppose we are interested in the ATE in percentage points for the treatment group, which consists of $G$ sub-treatment groups. Thus, the target population $P$ is the population of treated units. Consider the model
\begin{equation}
\ln Y_{i}=X^{\prime}_{i}\beta+D^{\prime}_{i}\tau+\epsilon_{i},\label{eq:ols-1}
\end{equation}
where $D_{i}=(D^{(1)}_{i},\ldots,D^{(G)}_{i})^{\prime}$ is a vector of dummies for the sub-treatment groups. If the conditional independence assumption holds and treatment effects do not vary with $X$, then $\tau_{g}$ is the ATE in log points for sub-treatment group $g$. \footnote{Section \ref{sec:examples} of the Online Appendix provides details for this case, in which $E(\tau_{i}\mid D^{(g)}_{i}=1,X_{i})=\tau_{g}$. If instead $E(\tau_{i}\mid D^{(g)}_{i}=1,X_{i})$ varies with $X_{i}$, then OLS estimates of $\tau_{g}$ need not coincide with the subgroup average log-point effect. As illustrated by \citet{angrist_estimating_1998} in the single treatment case, when controls are saturated in $X_{i}$, the OLS treatment coefficient equals a convex weighted average of covariate-specific treatment effects. Thus, with controls, the target population $P$ may reflect the covariate weighting induced by the regression.}

Under standard regularity conditions, the OLS estimator $\hat{\tau}$ satisfies $\sqrt{N}\left(\hat{\tau}-\tau\right)\xrightarrow{d}\mathcal{N}(0_{G\times1},\bar{\Sigma}_{\tau})$ for some finite matrix $\bar{\Sigma}_{\tau}$, with the heteroskedasticity robust variance estimator $\hat{\bar{\Sigma}}_{\tau}$ satisfying $\hat{\bar{\Sigma}}_{\tau}\xrightarrow{p}\bar{\Sigma}_{\tau}$ (see, e.g., \citealp{hayashi_econometrics_2000}). For the estimation of $w$, let $T_{i}=\sum^{G}_{g=1}D^{(g)}_{i}$ be the treatment indicator, and suppose the probability of treatment is $p_{T}=E(T_{i})\in(0,1)$, and define $w_{g}=E(D^{(g)}_{i}=1|T_{i}=1)\in(0,1)$. Under random sampling, a natural and consistent estimator of $w_{g}$ is $\hat{w}_{g}=N_{g}/N_{T}$, where $N_{g}=\sum^{N}_{i=1}D^{(g)}_{i}$ is the sample size of subgroup $g$ and $N_{T}=\sum^{G}_{g=1}N_{g}$ is the total number of treated observations. Since $(N_{1},\ldots,N_{G})\sim Multinomial(N_{T},w)$ and $N_{T}/N\rightarrow p_{T}$, it follows that $\sqrt{N}\left(\hat{w}-w\right)\xrightarrow{d}\mathcal{N}(0_{G\times1},\bar{\Sigma}_{w})$, where $\bar{\Sigma}_{w}=p^{-1}_{T}\left(\diag(w)-ww^{\prime}\right)$ and its consistent estimator can be obtained by replacing $w$ with $\hat{w}$ and $p_{T}$ with $\hat{p}_{T}=N_{T}/N$ respectively. Moreover, it can be shown that $\left(\sqrt{N}(\hat{w}-w),\sqrt{N}(\hat{\tau}-\tau)\right)$ is jointly asymptotically normal with zero covariance. Consequently, $\sqrt{N}\left(\hat{\delta}-\delta\right)\xrightarrow{d}\mathcal{N}\left(0_{2G\times1},\bar{\Sigma}_{\delta}\right)$, where $\bar{\Sigma}_{\delta}=\diag\{\bar{\Sigma}_{w},\bar{\Sigma}_{\tau}\}$, so that Assumption \ref{assu:delta} holds in this setting. See Section \ref{sec:examples} of the Online Appendix for further details.

Another example is staggered difference-in-differences designs. Heterogeneity-robust estimators in these designs typically provide consistent and asymptotically normal estimates of subgroup-specific treatment effects and weights, thereby satisfying Assumption \ref{assu:delta}. \footnote{See, for example, Theorems 2 and 3 in \citet{callaway_differenceindifferences_2021} and Proposition 6 in \citet{sun_estimating_2021}.} The estimation and inference procedures for $\rho_{b}$ developed here therefore apply. Example 2 in Section \ref{sec:examples} of the Online Appendix illustrates this application to a semi-log difference-in-differences model with staggered adoption. It is a direct extension of model (\ref{eq:ols-1}) when subgroups are defined as cohort-time combinations and the target population $P$ may correspond to all treated units, a selected cohort, a selected event time, or a selected calendar time, depending on the aggregate estimand of interest.

\subsection{Partial Identification of $\bar{\rho}$ \label{subsec:Partial-Identification}}

This section applies results from \citet{fan_partial_2017} to derive sharp bounds for $\bar{\rho}$ under unrestricted treatment-effect heterogeneity. The analysis follows directly from Theorem 3.2 of \citet{fan_partial_2017}, which delivers sharp bounds when the relevant conditional marginal distributions of the potential outcomes and the distribution of the conditioning covariates are identified. I take these distributions as given and focus on the resulting bound formulas for $\bar{\rho}$.

Assume that the two potential outcomes $Y_{1},Y_{0}\in(0,\infty)$ with $E(Y_{1})<\infty$. Observe that the function $\mu(Y_{1},Y_{0})=(Y_{1}-Y_{0})/Y_{0}$ is right-continuous and strictly sub-modular as for $Y^{\prime}_{1}>Y_{1}$ and $Y^{\prime}_{0}>Y_{0}$, we have 
\[
\frac{Y_{1}}{Y_{0}}+\frac{Y^{\prime}_{1}}{Y^{\prime}_{0}}-\frac{Y^{\prime}_{1}}{Y_{0}}-\frac{Y_{1}}{Y^{\prime}_{0}}=\frac{Y_{1}-Y^{\prime}_{1}}{Y_{0}}+\frac{Y^{\prime}_{1}-Y_{1}}{Y^{\prime}_{0}}=-\frac{(Y^{\prime}_{1}-Y_{1})(Y^{\prime}_{0}-Y_{0})}{Y_{0}Y^{\prime}_{0}}<0.
\]
With $E\mu(Y_{1},\bar{y}_{0})$ finite for any $\bar{y}_{0}>0$ and $E\mu(0,Y_{0})=-1$ also finite, Theorem 2 of \citet{cambanis_inequalities_1976} applies. Let $F_{1o}(Y_{1})$ and $F_{0o}(Y_{0})$ denote the marginal cumulative distribution functions (CDFs) of $Y_{1}$ and $Y_{0}$ respectively. The sharp bounds for $\mu(Y_{1},Y_{0})=(Y_{1}-Y_{0})/Y_{0}$ correspond to the Fr\'echet--Hoeffding bounds, with the lower bound $\rho_{L}=\int^{1}_{0}\left[F^{-1}_{1o}(u)/F^{-1}_{0o}(u)\right]du-1,$ and the upper bound $\rho_{U}=\int^{1}_{0}\left[F^{-1}_{1o}(1-u)/F^{-1}_{0o}(u)\right]du-1$, provided both integrals exist and at least one is finite.

The bounds can be tightened under the conditional independence assumption. Let $X$ denote observed covariates with support $\mathcal{X}\subset\mathcal{R}^{k_{x}}$. Under the assumption that $Y_{1},Y_{0}$ are jointly independent of treatment $T$ conditional on $X$ and that $0<E(T=1|X=x)<1$ for all $x\in\mathcal{X}$, we can identify the conditional marginal distributions of $Y_{1}$ and $Y_{0}$, denoted as $F_{1o}(Y_{1}|X)$ and $F_{0o}(Y_{0}|X)$ respectively. Following Theorem 3.2 of \citet{fan_partial_2017}, the sharp lower and upper bounds of $E\left[(Y_{1}-Y_{0})/Y_{0}\right]$ become 
\[
\rho^{\ast}_{L}=E\left[\int^{1}_{0}\frac{F^{-1}_{1o}(u|X)}{F^{-1}_{0o}(u|X)}du\right]-1,
\]
 and 
\begin{eqnarray*}
\rho^{\ast}_{U} & = & E\left[\int^{1}_{0}\frac{F^{-1}_{1o}(1-u|X)}{F^{-1}_{0o}(u|X)}du\right]-1,
\end{eqnarray*}
respectively, if $\rho^{\ast}_{L}$ and $\rho^{\ast}_{U}$ both exist and at least one of them is finite.

\section{Monte Carlo Experiment\label{sec:MC}}

This section reports Monte Carlo experiments that study the finite-sample performance of the proposed estimator $\hat{\rho}_{b}$ and the associated inference procedure in semi-log regression models. The simulation design follows the framework introduced at the end of Section \ref{sec:Estimation-and-Inference} and described in detail in Section \ref{sec:examples} of the Online Appendix.

I compare the performance of $\hat{\rho}_{b}$ with that of $\hat{\bar{\tau}}$ and $\hat{\rho}_{a}$, where $\hat{\bar{\tau}}=\sum^{G}_{g=1}\hat{w}_{g}\hat{\tau}_{g}$ is the estimator of $\bar{\tau}$, and $\hat{\rho}_{a}=\exp(\hat{\bar{\tau}})-1$ is the estimator of $\rho_{a}$. Under Assumption \ref{assu:delta}, $\hat{\bar{\tau}}\xrightarrow{p}\bar{\tau}$ and $\hat{\rho}_{a}\xrightarrow{p}\rho_{a}$ by the continuous mapping theorem. Standard delta-method calculations yield $\sqrt{N}(\hat{\bar{\tau}}-\bar{\tau})\xrightarrow{d}\mathcal{N}(0,\bar{\sigma}^{2}_{\bar{\tau}})$, where $\bar{\sigma}^{2}_{\bar{\tau}}=\nabla_{\delta}\bar{\tau}^{\prime}\bar{\Sigma}_{\delta}\nabla_{\delta}\bar{\tau}$ and $\nabla_{\delta}\bar{\tau}=(\tau^{\prime},w^{\prime})^{\prime}$. Similarly, $\sqrt{N}(\hat{\rho}_{a}-\rho_{a})\xrightarrow{d}\mathcal{N}(0,\bar{\sigma}^{2}_{a})$, where $\bar{\sigma}^{2}_{a}=\exp(2\bar{\tau})\bar{\sigma}^{2}_{\bar{\tau}}$. To assess the finite-sample performance of the inference method, I report two-sided $z$-tests for $H_{0}:\bar{\tau}=\rho_{0}$ based on $z_{\tau}=(\hat{\bar{\tau}}-\rho_{0})/\sigma_{\bar{\tau}}$, where $\sigma_{\bar{\tau}}=\bar{\sigma}_{\bar{\tau}}/\sqrt{N}$ is estimated by replacing each parameter with its corresponding estimate. Because $\rho_{a}=\exp(\bar{\tau})-1$, testing $H_{0}:\rho_{a}=\rho_{0}$ is equivalent to testing $H_{0}:\bar{\tau}=\ln(\rho_{0}+1)$. The corresponding statistic is $z_{a}=(\hat{\bar{\tau}}-\ln(\rho_{0}+1))/\sigma_{\bar{\tau}}$.

Because $\bar{\rho}\geqslant\rho_{b}\geqslant\rho_{a}\geqslant\bar{\tau}$, differences between $\hat{\rho}_{b}$ and the conventional estimators $\hat{\rho}_{a}$ and $\hat{\bar{\tau}}$ capture the extent to which $\hat{\rho}_{b}$ is closer to the target $\bar{\rho}$, regardless of whether $\bar{\rho}=\rho_{b}$. For ease of interpretation, I impose constant treatment effects within subgroups so that $\rho_{b}=\bar{\rho}$.

The simulations focus on the average treatment effect on the treated (ATT) in percentage points with $G=4$ sub-treatment groups. Data are generated according to the following semi-log model: $\ln Y_{i}=1+X_{i}+\sum^{4}_{g=1}D^{(g)}_{i}\tau_{g}+\epsilon_{i}$, where $D^{(g)}_{i}$ is the dummy for sub-treatment group $g$. The treatment probability is $p_{T}=0.8$, and each sub-treatment group has equal weight $w_{g}=0.25$ within the treated population. Consequently, the control group and each sub-treatment group each comprise $20\%$ of the total population. Observations are randomly assigned to groups according to these probabilities. Sample size $N$ is selected from \{20, 50, 100, 200, 500, 1000, 2000, 5000, $10^{4}$, $10^{5}$, $10^{6}$\}. The covariate $X_{i}$ and the error term $\epsilon_{i}$ are each i.i.d. $\mathcal{N}(0,1)$. Results using skew-normal errors are analogous and presented in Section \ref{sec:app_mc} of the Online Appendix. I examine two cases: a large heterogeneity case with $\tau=(\ln(0.68),\ln(0.84),\ln(1.16),\ln(1.32))^{\prime}$, implying $100(\rho_{1},\rho_{2},\rho_{3},\rho_{4})=(-32,-16,16,32)$, and a small heterogeneity case with $\tau=(\ln(0.92),\ln(0.96),\ln(1.04),\ln(1.08))^{\prime}$, implying $100(\rho_{1},\rho_{2},\rho_{3},\rho_{4})=(-8,-4,4,8)$. In both cases, the true value of $\bar{\rho}=\rho_{b}$ equals $0\%$. I estimate $\hat{\tau}$, $\hat{w}$, and $\hat{\bar{\Sigma}}_{\delta}$ as described at the end of Section \ref{sec:Estimation-and-Inference}, and then construct $\hat{\bar{\tau}}$, $\hat{\rho}_{a}$ and $\hat{\rho}_{b}$ together with their associated $z$-statistics.

\begin{table}[hp]
\small
\caption{Monte Carlo Results for Estimation and Inference on $\rho_b$: Normal Errors}
\label{tab:monteS0truew0}
\begin{center}
\hspace*{-10mm}
\begin{threeparttable}
\begin{tabular}{cccccccccccccc}
\hline \hline  & \multicolumn{6}{c}{Small Heterogeneity} & & \multicolumn{6}{c}{ Large Heterogeneity} \\ \cline{2-7} \cline{9-14}
& \multicolumn{3}{c}{Estimators} & & \multicolumn{2}{c}{ERR} & & \multicolumn{3}{c}{Estimators} & & \multicolumn{2}{c}{ERR} \\ \cline{2-4} \cline{6-7} \cline{9-11} \cline{13-14}
$ N $ & $ \hat{\bar{\tau}} $ & $ \hat{\rho}_a $  & $ \hat{\rho}_b $ & & $ z_{\tau} $ & $ z_{\rho} $ & &  $ \hat{\bar{\tau}} $ & $ \hat{\rho}_a $  & $ \hat{\rho}_b $ & & $ z_{\tau} $ & $ z_{\rho} $ \\ \hline
 \multicolumn{14}{c}{\textit{True Values}} \\
      & $ -0.201 $ & $ -0.200 $  & $ 0 $ & & $ 5 $ & $ 5 $ & & $ -3.349 $ & $ -3.294 $  & $ 0 $ & & $ 5 $ & $ 5 $ \\
 \multicolumn{14}{c}{\textit{Estimates}} \\20 & -0.44 & 20.63 & 33.99 &  & 6.44 & 7.84 &  & -3.70 & 16.94 & 33.93 &  & 6.38 & 7.93  \\
 & (61.68) & (85.37) & (96.03) &  &  &  &  & (61.80) & (84.02) & (97.53) &  &  &   \\
50 & -0.12 & 7.01 & 11.45 &  & 5.35 & 5.87 &  & -3.33 & 3.76 & 11.44 &  & 5.59 & 6.12  \\
 & (37.09) & (41.61) & (43.48) &  &  &  &  & (37.40) & (40.69) & (44.06) &  &  &   \\
100 & -0.23 & 3.09 & 5.28 &  & 5.15 & 5.44 &  & -3.43 & -0.10 & 5.24 &  & 5.48 & 5.54  \\
 & (25.56) & (26.84) & (27.47) &  &  &  &  & (25.80) & (26.26) & (27.86) &  &  &   \\
200 & -0.13 & 1.48 & 2.65 &  & 5.14 & 5.14 &  & -3.38 & -1.75 & 2.55 &  & 5.49 & 5.18  \\
 & (17.90) & (18.31) & (18.55) &  &  &  &  & (17.96) & (17.78) & (18.68) &  &  &   \\
500 & -0.20 & 0.44 & 1.02 &  & 5.16 & 5.17 &  & -3.37 & -2.69 & 1.00 &  & 6.12 & 5.18  \\
 & (11.28) & (11.36) & (11.43) &  &  &  &  & (11.34) & (11.07) & (11.56) &  &  &   \\
1000 & -0.22 & 0.10 & 0.48 &  & 5.06 & 5.10 &  & -3.37 & -3.01 & 0.48 &  & 7.14 & 5.03  \\
 & (7.94) & (7.95) & (7.99) &  &  &  &  & (7.98) & (7.76) & (8.08) &  &  &   \\
2000 & -0.22 & -0.06 & 0.23 &  & 5.00 & 5.00 &  & -3.32 & -3.11 & 0.28 &  & 9.14 & 5.04  \\
 & (5.59) & (5.59) & (5.61) &  &  &  &  & (5.66) & (5.48) & (5.71) &  &  &   \\
5000 & -0.19 & -0.13 & 0.11 &  & 5.11 & 5.09 &  & -3.35 & -3.24 & 0.10 &  & 15.58 & 4.98  \\
 & (3.54) & (3.53) & (3.54) &  &  &  &  & (3.56) & (3.44) & (3.58) &  &  &   \\
$10^4$ & -0.20 & -0.17 & 0.04 &  & 4.94 & 4.88 &  & -3.35 & -3.26 & 0.05 &  & 26.29 & 4.89  \\
 & (2.48) & (2.48) & (2.49) &  &  &  &  & (2.50) & (2.42) & (2.52) &  &  &   \\
$10^5$ & -0.20 & -0.19 & 0.01 &  & 5.70 & 4.95 &  & -3.35 & -3.29 & 0.01 &  & 98.76 & 4.93  \\
 & (0.79) & (0.79) & (0.79) &  &  &  &  & (0.79) & (0.77) & (0.80) &  &  &   \\
$10^6$ & -0.20 & -0.20 & 0.00 &  & 12.50 & 4.98 &  & -3.35 & -3.29 & 0.00 &  & 100.00 & 4.94  \\
 & (0.25) & (0.25) & (0.25) &  &  &  &  & (0.25) & (0.24) & (0.25) &  &  &   \\
\hline \hline \end{tabular}
\begin{tablenotes}
\footnotesize \item 1. Entries report Monte Carlo means and standard deviations (in parentheses) for $\hat{\bar{\tau}}$, $\hat{\rho}_a$, $\hat{\rho}_b$, based on $ 10^5$ replications for each sample size $N$. 
Also reported are empirical rejection rates (ERRs, in percent) at the 5\% nominal level for two-sided tests of $H_0: \bar{\tau}=0$ using $z_\tau$, and of $H_0: \rho_b=0$ using $z_{\rho}$. 
 \item 2. The left panel uses $\tau=(\ln(0.92),\ln(0.96),\ln(1.04),\ln(1.08))'$, and the right panel uses $\tau=(\ln(0.68),\ln(0.84),\ln(1.16),\ln(1.32))'$. The error terms $ \epsilon $ are i.i.d. $ \mathcal N(0,1) $. 
 \item 3. All values are multiplied by 100. True parameter values are in the first row.
\end{tablenotes}
\end{threeparttable}
\end{center}
\end{table}

The left and right panels of Table \ref{tab:monteS0truew0} present results for small and large heterogeneity respectively, with all values multiplied by 100. The first three columns in each panel report the Monte Carlo means and standard deviations (in parentheses) of $\hat{\bar{\tau}}$, $\hat{\rho}_{a}$ and $\hat{\rho}_{b}$ across 100,000 repetitions, with true values provided at the top of each panel. The last two columns in each panel report empirical rejection rates of two-sided $z$-tests for $H_{0}:\bar{\tau}=0$ using $z_{\tau}$ and for $H_{0}:\rho_{b}=0$ using $z_{\rho}$ at the 5\% level. The empirical rejection rates equal one minus the coverage rate for $0$ of the confidence intervals for $\bar{\tau}$ and $\rho_{b}$ respectively. The $z$-test for $\rho_{a}=0$ is omitted, because $\rho_{a}=0$ is equivalent to $\bar{\tau}=0$ given that $\rho_{a}=\exp(\bar{\tau})-1$.

The proposed estimator $\hat{\rho}_{b}$ and associated $z$-test perform well even for modest sample sizes. For both designs with large and small treatment effect heterogeneity, when $N\geqslant500$, that is when each subgroup has $100$ observations, the bias of $\hat{\rho}_{b}$ is below 1 percentage point. The empirical rejection rate of the test based on $z_{\rho}$ is close to nominal, ranging from 4.9\% to 5.5\% when $N\geqslant100$.

The estimators $\hat{\bar{\tau}}$ and $\hat{\rho}_{a}$ provide reasonable approximations for $\rho_{b}$ in the case of small heterogeneity, as the true values of $\bar{\tau}$ and $\rho_{a}$ are around $-0.2\%$, close to $\rho_{b}=0$. The approximation bias converges to $-0.2\%$ as the sample size increases. Since $\bar{\tau}$ is close to, but not exactly, zero in this design, tests of $H_{0}:\bar{\tau}=0$ using $z_{\tau}$ reject at rates close to the $5\%$ nominal level for $N\leqslant10^{5}$ but reject more often in very large samples, e.g., $12.5\%$ when $N=10^{6}$. Intuitively, as $N$ grows, $\sigma_{\bar{\tau}}=\bar{\sigma}_{\bar{\tau}}/\sqrt{N}$ shrinks, causing the relative bias $(\rho_{b}-\bar{\tau})/\sigma_{\bar{\tau}}$ to grow.

In the large heterogeneity design, the approximation bias of $\hat{\bar{\tau}}$ and $\hat{\rho}_{a}$ is much larger and converges to $-3.3\%$. Inference based on $z_{\tau}$ can therefore differ substantially from inference based on $z_{\rho}$. For example, the rejection rate based on $z_{\tau}$ is 9.14\% when $N=2000$. Overall, the simulations show that reporting $\rho_{b}$ can matter both for point estimates and for inference, especially when dealing with large treatment effect heterogeneity or very large sample sizes.

\section{Empirical Applications\label{sec:Empirical-Applications}}

This section uses two empirical applications to illustrate the proposed estimation and inference procedures. The applications further assess how $\hat{\rho}_{b}$ compares with conventional estimators $\hat{\bar{\tau}}$ and $\hat{\rho}_{a}$ in practice.

\subsection{Exporting and Firm Performance}

The first application uses a semi-log regression model to analyze the randomized experiment in \citet{atkin_exporting_2017}, which examines the causal impact of exporting on the performance of small rug producers in Egypt.

\citet{atkin_exporting_2017} recruited two samples of firms satisfying specific criteria. Each sample was divided into strata based on rug type and loom size. Within each stratum, firms were randomly assigned to the treatment group and offered access to export orders, though not all took up. This replication focuses on the first three follow-up rounds of the joint sample of 219 firms producing duble rugs. The sample comprises 79 firms from 4 strata in sample 1, and 140 firms from 4 strata in sample 2. Of the 74 firms in the treatment group, 47 took up the export opportunity.

Given randomization within strata, it is reasonable to expect treatment effects to differ across strata. I therefore replace the single treatment dummy with dummies for eight sub-treatment groups, modifying Eq. (1) in \citet{atkin_exporting_2017} to:

\begin{equation}
\ln(Y_{igt})=\alpha+\sum^{8}_{g=1}\tau_{g}D^{(g)}_{i}+\gamma\ln(Y_{ig0})+\delta_{g}+\theta_{t}+\epsilon_{igt},\label{eq:atk}
\end{equation}
where $Y_{igt}$ is the outcome of firm $i$ in stratum $g$ in round $t$, $Y_{ig0}$ is the baseline outcome, and $\delta_{g}$ and $\theta_{t}$ capture stratum and round fixed effects. Outcomes include various measures of profits in Table V and determinants of profits like price and output in Table VI of \citet{atkin_exporting_2017}. The sub-treatment group dummy $D^{(g)}_{i}=1$ if firm $i$ is in stratum $g$ and assigned to the treatment group. Since not all treatment firms took up the treatment, $\tau_{g}$ represents the average intent to treat (ITT) effect in log points for sub-treatment group $g$. Standard errors are clustered at the firm level.

  \begin{table}[h!] \begin{center} \small \caption{Replication of Atkin et al. (2017)} \label{tab:Atk}  \begin{threeparttable}  \begin{tabular}{cccc} \hline \hline \hspace{10pt} Outcome \hspace{10pt} &\hspace{25pt} $\hat{\bar{\tau}} \hspace{25pt} $ & $ \hspace{25pt} \hat{\rho}_a \hspace{25pt} $ & $\hspace{25pt} \hat{\rho}_b \hspace{25pt} $ \\ \hline 
\multicolumn{4}{c}{\textit{Table V: Impact of Exporting on Firm Profits}} \\
Direct&22.2(6.2)&24.9(7.8)&26.5(7.9)\\
Reported&19.3(6.7)&21.3(8.1)&22.7(8.7)\\
Constructed&16.3(6.9)&17.7(8.1)&19.9(8.4)\\
Hypothetical&36.0(11.2)&43.3(16.0)&53.4(17.5)\\
\multicolumn{4}{c}{\textit{Table V: Impact of Exporting on Firm Profits (per owner hour)}} \\
Direct&16.8(5.6)&18.3(6.7)&20.3(7.2)\\
Reported&15.6(6.1)&16.9(7.1)&18.6(8.1)\\
Constructed&13.2(6.1)&14.1(7.0)&15.7(7.4)\\
Hypothetical&26.5(7.0)&30.3(9.1)&31.5(10.0)\\
\multicolumn{4}{c}{\textit{Table VI: Impact of Exporting on Components of Profits}} \\
Output price&39.4(10.7)&48.3(15.8)&57.5(16.8)\\
Output&-24.8(9.4)&-22.0(7.3)&-19.6(7.5)\\
Hours worked&4.2(2.3)&4.3(2.5)&4.8(2.4)\\
No. of looms&-1.1(4.1)&-1.1(4.1)&-0.8(4.3)\\
Warp thread balls&12.8(5.0)&13.7(5.7)&15.6(6.0)\\
 \hline \hline \end{tabular} \begin{tablenotes} \footnotesize
\item 1. Estimates and standard errors (in parentheses) are reported for $\hat{\bar{\tau}}$, $\hat{\rho}_a=\exp(\hat{\bar{\tau}})-1$ and $\hat{\rho}_b$. All values are multiplied by 100. \item 2. Panel 1 replicates Columns (1), (3), (5), and (7) from Panel A of Table V in Atkin et al. (2017). Panel 2 replicates Panel B of the same table. Panel 3 replicates Columns (1), (3), (5), (9), and (11) of Table VI in Atkin et al. (2017).
\end{tablenotes} \end{threeparttable} \end{center} \end{table}

Table \ref{tab:Atk} presents estimates and standard errors (in parentheses) for $\hat{\bar{\tau}}$, $\hat{\rho}_{a}=\exp(\hat{\bar{\tau}})-1$, and $\hat{\rho}_{b}=\sum\hat{w}_{g}\exp(\hat{\tau}_{g})-1$. The first panel replicates Columns (1), (3), (5), and (7) of Table V Panel A in \citet{atkin_exporting_2017}, showing the impact on various profit measures. The second panel replicates Panel B, showing the impact on profits per owner hour. The third panel replicates Columns (1), (3), (5), (9), and (11) of Table VI, showing the impact on various determinants of profits.

For most outcomes, $\hat{\bar{\tau}}$ differs substantially from $\hat{\rho}_{b}$. The difference between $\hat{\rho}_{a}$ and $\hat{\rho}_{b}$ is about 1 to 2 percentage points for many outcomes. For example, the impact of exporting on direct profits per owner hour is $18.3\%$ (95\% CI: $[5.3\%,31.4\%]$) using $\hat{\rho}_{a}$,\footnote{The confidence interval for $\rho_{a}$ here is $CI_{a}=\left[\exp\left(\hat{\bar{\tau}}+\hat{\sigma}_{\bar{\tau}}z_{\alpha/2}\right)-1,\exp\left(\hat{\bar{\tau}}-\hat{\sigma}_{\bar{\tau}}z_{\alpha/2}\right)-1\right]$, where $z_{\alpha/2}$ is the $\alpha/2$ quantile of the standard normal distribution.} and $20.3\%$ (95\% CI: $[6.2\%,34.5\%]$) using $\hat{\rho}_{b}$, differing by 2 percentage points. The differences are larger for hypothetical profits and output price, e.g., $\hat{\rho}_{a}=43.3\%$ (95\% CI: $[11.9\%,74.6\%]$) and $\hat{\rho}_{b}=53.4\%$ (95\% CI: $[19.0\%,87.8\%]$) for hypothetical profits.

Overall, the difference between $\hat{\rho}_{b}$ and $\hat{\rho}_{a}$ is modest for most outcomes in this application. However, the results also show that when treatment effect heterogeneity is large, meaningful gaps can arise between conventional estimates and $\hat{\rho}_{b}$.

\subsection{Impact of Water and Sewerage Infrastructures on Child Mortality}

The second application illustrates the applicability of my methodology to staggered difference-in-differences designs. It extends the analysis in Table 2 of \citet{alsan_watersheds_2019}, examining the impact of water and sewerage infrastructures on child mortality in the United States.

\citet{alsan_watersheds_2019} explore the staggered rollout of water and sewerage infrastructures across municipalities in Massachusetts between 1892 and 1903. They utilize panel data of 60 Massachusetts municipalities from 1880 to 1920. Their main finding is that neither water nor sewerage infrastructure alone affected child mortality, but the combination of both systems reduced child mortality by $0.266$ log points.

Based on this finding, I define treatment as the joint presence of both water and sewerage infrastructures, while controlling for the separate effects of having only one system. A municipality is in cohort $c$ if it first obtained both sewerage and safe water infrastructures in year $c$. Table A1 in \citet{alsan_watersheds_2019} documents the intervention year for each municipality.

I estimate $\tau$ using the following model 
\begin{equation}
\ln(Y_{it})=\alpha_{i}+\beta_{t}+X^{\prime}_{it}\gamma+\sum_{c\neq\infty}\sum_{r\neq-1}D_{it}(c,r)\tau(c,r)+\epsilon_{it},\label{eq:stagdid-1}
\end{equation}
where $Y_{it}$ is the under-5 mortality rate in municipality $i$ in year $t$. The treatment indicator $D_{it}(c,r)$ equals 1 if municipality $i$ belongs to cohort $c$ and year $t$ is its $r$-th post-treatment period.\footnote{Following \citet{alsan_watersheds_2019}, event time is grouped into 2-year bins except for period $-1$. For example, $r=0$ for event times $0$ and 1, $r=1$ for event times $2$ and $3$, $r=5$ for event times $10$ and above. Similarly, $r=-2$ for event times $-3$ and $-2$, and $r=-6$ for event times $\leqslant-10$.} The specification includes municipality fixed effects $\alpha_{i}$ and year fixed effects $\beta_{t}$. Following \citet{alsan_watersheds_2019}, the model controls for time-varying covariates $X_{it}$,\footnote{The use of a TWFE regression with time-varying covariates in difference-in-differences designs requires a conditional parallel trends assumption. The covariates $X_{it}$ should be strictly exogenous and should not themselves be affected by treatment, so as to avoid the \textquotedblleft bad control\textquotedblright{} problem. In addition, the conditional expectation function of the untreated potential outcome must be linear in $X_{it}$ \citep{wooldridge_twoway_2025}.} including (1) $Wateralone$, indicator for having only a water system but no sewerage system in year $t$, (2) $Seweragealone$, indicator for having only a sewerage system but no safe water system; (3) ``log of population density, percentage foreign-born, percentage male, and the percentage of females employed in manufacturing'' and ``municipality specific linear trends''. I estimate $\tau(c,r)$ by OLS and cluster the standard errors at the municipality level following \citet{alsan_watersheds_2019}. Weights for aggregate effects use sample shares, as described in detail in Example 2 of Section \ref{sec:examples} of the Online Appendix. I set $\bar{\Sigma}_{w}=0$ since the sample represents the complete population of Massachusetts municipalities.

  \begin{table}[h!] \begin{center} \small \caption{Replication of Table 2 in Alsan and Goldin (2019)} \label{tab:AG2019}  \begin{threeparttable}  \begin{tabular}{cccc} \hline \hline &\hspace{25pt} $\hat{\bar{\tau}} \hspace{25pt} $ & $ \hspace{25pt} \hat{\rho}_a \hspace{25pt} $ & $\hspace{25pt} \hat{\rho}_b \hspace{25pt} $ \\ \hline 
\multicolumn{4}{c}{\textit{ATT for All Treated Units}} \\
&-51.4(10.0)&-40.2(6.0)&-37.8(6.0)\\
\multicolumn{4}{c}{\textit{ATT by Event Time}} \\
0&-24.0(7.6)&-21.3(6.0)&-19.7(6.0)\\
1&-30.0(8.5)&-25.9(6.3)&-20.0(6.7)\\
2&-39.0(9.2)&-32.3(6.2)&-29.2(6.6)\\
3&-39.2(10.3)&-32.5(7.0)&-30.0(7.1)\\
4&-54.7(11.3)&-42.1(6.5)&-41.0(6.7)\\
5&-63.2(11.6)&-46.8(6.1)&-46.1(6.2)\\
\multicolumn{4}{c}{\textit{ATT by Cohort}} \\
1898&-49.1(11.8)&-38.8(7.2)&-37.7(7.2)\\
1899&-52.3(33.0)&-40.7(19.5)&-40.2(19.1)\\
1901&-7.0(8.5)&-6.7(7.9)&-4.9(8.0)\\
1902&-96.6(8.5)&-61.9(3.2)&-61.3(3.2)\\
1903&-28.9(10.7)&-25.1(8.0)&-19.5(8.3)\\
 \hline \hline \end{tabular} \begin{tablenotes} \footnotesize
\item 1. This table replicates Table 2 of Alsan and Goldin (2019), who study the impact of water and sewerage infrastructure on child mortality rates in Massachusetts in the 1890s. \item 2. Estimates and standard errors (in parentheses) are reported for $\hat{\bar{\tau}}$, $\hat{\rho}_a=\exp(\hat{\bar{\tau}})-1$ and $\hat{\rho}_b$. All values are multiplied by 100. \item 3. Panel 1 presents ATT for all treated units. Panel 2 presents the ATT for event time $ 0,1,\ldots,5 $. Panel 3 presents the ATT for each cohort.
\end{tablenotes} \end{threeparttable} \end{center} \end{table}

Results are summarized in Table \ref{tab:AG2019}, with all values multiplied by 100. From top to bottom are estimates and standard errors of $\hat{\bar{\tau}}$, $\hat{\rho}_{a}=\exp(\hat{\bar{\tau}})-1$, and $\hat{\rho}_{b}$ for all treated units, each event time, and each cohort. Due to their large absolute values, $\hat{\bar{\tau}}$ differs dramatically from $\hat{\rho}_{a}$ and $\hat{\rho}_{b}$ throughout. The conventional estimator $\hat{\rho}_{a}$ and the proposed estimator $\hat{\rho}_{b}$ also differ meaningfully in some cases, with gaps reaching up to $5.9$ percentage points. For all treated units, $\hat{\rho}_{a}$ ($-40.2\%$) overstates the mortality reduction compared to the proposed estimator $\hat{\rho}_{b}$ ($-37.8\%$), a difference of $2.4$ percentage points. The gap between the two estimators varies across event times: relatively small for later periods ($\hat{\rho}_{a}=-46.8\%$ vs. $\hat{\rho}_{b}=-46.1\%$ at event time 5) but substantial for early periods, particularly event time 1 where $\hat{\rho}_{a}=-25.9\%$ and $\hat{\rho}_{b}=-20.0\%$. Cohort-specific estimates show even greater variation: while cohorts 1898, 1899, and 1902 exhibit relatively small differences between $\hat{\rho}_{a}$ and $\hat{\rho}_{b}$, cohort 1903 shows a substantial gap, with $\hat{\rho}_{a}=-25.1\%$ and $\hat{\rho}_{b}=-19.5\%$.

As in the first application, $\hat{\rho}_{a}$ is often close to $\hat{\rho}_{b}$, so conventional estimators can be adequate in many cases. However, when treatment effects vary substantially across subgroups, the resulting gap can be economically meaningful. Because the magnitude of this gap is difficult to gauge ex ante, reporting $\hat{\rho}_{b}$ therefore provides a practical diagnostic and a way to account for observable heterogeneity when interpreting ATEs in percentage points.

Using the imputation estimator of \citet{borusyak_revisiting_2024}, I obtain $\hat{\rho}^{imp}_{b}=-9.9\%$ for all treated units. As discussed in Remark \ref{rem:imputation}, this realized value is not a consistent estimator of the ATE in percentage points for all treated units, although it has an upper-bound interpretation in expectation in this setting. See Section \ref{subsec:Imputation-AG} of the Online Appendix for details.

\section{Conclusion\label{sec:Conclusion}}

This paper highlights the importance of correctly estimating and interpreting ATEs in percentage points when treatment effects are heterogeneous. Differences between ATEs in log points and in percentage points can be substantial, especially when treatment effects are large or vary significantly. Failing to account for these differences may lead to misinterpretation of results and potentially misguided policy recommendations.

My proposed methods provide researchers with tools for improved estimation and inference of ATEs in percentage points in the presence of treatment effect heterogeneity. By accounting for heterogeneity across observable subgroups, the methods yield point identification when treatment effects are constant within subgroups and informative lower bounds otherwise. The methods can be applied to a variety of settings such as semi-log regression models. They are particularly relevant for research designs such as staggered difference-in-differences models, where treatment effect heterogeneity may be present.

By applying the methods to empirical studies on how exporting affects firm productivity and how the combination of water and sewerage infrastructure affects children's mortality, I demonstrate how accounting for heterogeneity can affect the interpretation of ATEs in percentage points in practice. As empirical studies continue to grapple with complex treatment effect patterns, tools like those presented in this paper will be useful for accurate estimation and inference.

\subsection*{Declaration of competing interest}

The author declares that she has no known competing financial interests or personal relationships that could have appeared to influence the work reported in this paper.

\subsection*{Funding}

This work was supported by the Young Scientists Fund of the National Natural Science Foundation of China (Grant No. 72403212), the Young Scientists Fund of the Fujian Natural Science Foundation (Grant No. 2023J05010) and the General Program of the National Natural Science Foundation of China (Grant Nos. 72573138 and 72573135).

\subsection*{Declaration of generative AI and AI-assisted technologies in the manuscript preparation process}

During the preparation of this work, the author used Claude and ChatGPT in order to improve the readability and language of the paper, as well as to assist with coding in Stata and R. After using these tools/services, the author reviewed and edited the content as needed and takes full responsibility for the content of the published article.

\subsection*{Acknowledgments}

I am grateful to the editor, the associate editor, and two anonymous referees for their constructive comments and suggestions that substantially improved this paper. I thank Ingmar Prucha, Guido Kuersteiner, and participants at the NSFC 2025 Young Scientists Fund (Category C, Economics) Awardees\textquoteright{} Meeting, the 10th Young Econometricians in Asia-Pacific (YEAP) Annual Meeting, and the lunch seminar at the Department of Public Economics, Xiamen University, for valuable comments. This work was supported by the High Performance Computing Platform of the Key Laboratory of Econometrics (Xiamen University), Ministry of Education. All remaining errors are my own.

\bibliographystyle{elsarticle-harv}
\bibliography{loglinear3}
\pagebreak{}

\section*{Online Appendix for ``Estimation and Inference on Average Treatment Effect in Percentage Points under Heterogeneity''}

\setcounter{section}{0}\setcounter{page}{1}\setcounter{table}{0}\setcounter{equation}{0}\setcounter{lem}{0}\setcounter{assumption}{0}\setcounter{thm}{0}\setcounter{footnote}{0}\setcounter{rem}{0}\setcounter{cor}{0}\setcounter{prop}{0}

\renewcommand{\theequation}{O.\arabic{equation}}
\renewcommand{\thetable}{O.\arabic{table}}
\renewcommand{\thesection}{O.\arabic{section}}
\renewcommand{\theassumption}{O.\arabic{assumption}}
\renewcommand{\thelem}{O.\arabic{lem}}
\renewcommand{\thethm}{O.\arabic{thm}}
\renewcommand{\therem}{O.\arabic{rem}}
\renewcommand{\thecor}{O.\arabic{cor}}

This online appendix includes four sections. Section \ref{sec:examples} demonstrates how Assumption \ref{assu:delta} in the main text can be satisfied in semi-log regression models and staggered difference-in-differences models. Section \ref{sec:normal heteogeneity} discusses estimation and inference of the ATE in percentage points when potential outcomes are jointly normal. Section \ref{sec:rho_c} presents a bias-corrected estimator for $\rho_{b}$. Section \ref{sec:app_mc} reports additional Monte Carlo simulations and empirical results.

\section{Estimators Satisfying Assumption \ref{assu:delta}\label{sec:examples}}

\subsection{Example 1: Semi-log Regression Model}

This example shows how $\hat{w}$ and $\hat{\tau}$ satisfying Assumption \ref{assu:delta} in the main text can be obtained as OLS estimators in semi-log regression models, applicable in research designs such as (conditional) randomized controlled trials.

The objective is to estimate the average treatment effect on the treated (ATT) in percentage points. The population of interest is the entire treatment group, comprising $G$ sub-treatment groups. In this setting, $D^{(g)}_{i}=1$ if $i$ is in sub-treatment group $g$ and $0$ otherwise. The indicator of treatment is $T_{i}=\sum^{G}_{g=1}D^{(g)}_{i}=\mathbf{1}^{\prime}_{G}D_{i}$, where $D_{i}=(D^{(1)}_{i},\ldots,D^{(G)}_{i})^{\prime}$ and $\mathbf{1}_{G}$ is a $G\times1$ vector of ones.

I obtain $\hat{\tau}$ as the OLS estimator from a semi-log regression model. The observed log-transformed outcome is $\ln Y_{i}=\ln Y_{0i}(1-T_{i})+\ln Y_{1i}T_{i}=\ln Y_{0i}+T_{i}\tau_{i}$. Let $X$ denote a $k_{x}\times1$ vector of observed covariates including a constant term, and assume that $E\left(\ln Y_{0i}\left|X_{i},D_{i}\right.\right)=X^{\prime}_{i}\beta$. Assume that $E(\tau_{i}|D^{(g)}_{i}=1,X_{i})=\tau_{g}$ so $E(\tau_{i}|D_{i},X_{i})=D^{\prime}_{i}\tau$ and $E(T_{i}\tau_{i}|D_{i},X_{i})=D^{\prime}_{i}\tau$. A semi-log linear regression model that allows for treatment effect heterogeneity is specified as
\begin{equation}
\ln Y_{i}=X^{\prime}_{i}\beta+D^{\prime}_{i}\tau+\epsilon_{i},\label{eq:ols}
\end{equation}
where $\epsilon_{i}=\ln Y_{i}-E[\ln Y_{i}|X_{i},D_{i}]$ and by construction $E(\epsilon_{i}|X_{i},D_{i})=0$. \footnote{The specification can also be justified by other assumptions and alternative estimation approaches are available, see \citet{goldsmith-pinkham_contamination_2024} for a comprehensive discussion of the estimation of heterogeneous treatment effects in linear regressions.}
\begin{rem}
The specification in (\ref{eq:ols}) imposes no restriction on the heterogeneity of $\ln Y_{0i}$ and permits various group structures. The control group can be homogeneous or partitioned into $G$ subgroups corresponding to treatment subgroups. In the latter case, $X_{i}$ includes subgroup dummies, and $D^{(g)}_{i}$ represents subgroup-treatment interactions. For instance, consider a sample comprising $G$ cities, each containing both control and treatment units. Here, $X_{i}$ would include city dummies, and $D^{(g)}_{i}$ would be city-treatment interactions, capturing city-specific effects. In this specific scenario, the estimator of $\tau$ is analogous to the interaction-weighted estimator in \citet{gibbons_broken_2019}, differing only in the log-transformed outcome.

The matrix form of model (\ref{eq:ols}) is $\mathbf{lnY}=\mathbf{X}\beta+\mathbf{D}\tau+\epsilon,$ where $\mathbf{\mathbf{lnY}}=\left(\ln Y_{1},\ldots,\ln Y_{N}\right)^{\prime}$, $\boldsymbol{\mathbf{X}}=\left(X_{1},\ldots,X_{N}\right)^{\prime}$, $\mathbf{D}=(D_{1},\ldots,D_{N})^{\prime}$, and $\epsilon=\left(\epsilon_{1},\ldots,\epsilon_{N}\right)^{\prime}$. Let $\tilde{X}_{i}=\left(X^{\prime}_{i},D^{\prime}_{i}\right)^{\prime}$, and $\tilde{\mathbf{X}}=(\tilde{X}_{1},\ldots\tilde{X}_{N})^{\prime}=[\mathbf{X},\mathbf{D}]$. The OLS estimator of $(\beta^{\prime},\tau^{\prime})^{\prime}$ is given by $(\hat{\beta}^{\prime},\hat{\tau}^{\prime})^{\prime}=(\tilde{\mathbf{X}}^{\prime}\tilde{\mathbf{X}})^{-1}\tilde{\mathbf{X}}^{\prime}\mathbf{lnY}$. The heteroscedasticity robust estimator for the asymptotic variance of $\sqrt{N}(\hat{\beta}^{\prime},\hat{\tau}^{\prime})^{\prime}$ is 
\begin{equation}
\hat{\bar{\Sigma}}_{(\beta,\tau)}=N(\tilde{\mathbf{X}}^{\prime}\tilde{\mathbf{X}})^{-1}\left(\sum^{N}_{i=1}\tilde{X}_{i}\tilde{X}^{\prime}_{i}\hat{\epsilon}^{2}_{i}\right)(\tilde{\mathbf{X}}^{\prime}\tilde{\mathbf{X}})^{-1},\label{eq:est_Sigma_btau}
\end{equation}
where $\hat{\epsilon}_{i}$ is the OLS residual. Alternative covariance estimators, including cluster-robust variance estimators, also satisfy Assumption \ref{assu:delta} provided that they are consistent.

I next discuss the estimation of $w$. The population share of sub-treatment group $g$ in the whole treatment group is $w_{g}=E(D^{(g)}_{i}=1|T_{i}=1)$, where $T_{i}=\sum^{G}_{g=1}D^{(g)}_{i}$ is the indicator that $i$ belongs to the treatment group. Suppose we have a sample of $N$ individuals, with $N_{T}=\sum^{N}_{i=1}T_{i}$ in the treatment group. Define $N_{g}=\sum^{N}_{i=1}D^{(g)}_{i}$ as the size of sub-treatment group $g$ so that $N_{T}=\sum^{G}_{g=1}N_{g}$. The estimator of $w_{g}$ is $\hat{w}_{g}=N_{g}/N_{T}$ and of $w$ is $\hat{w}=(N_{1}/N_{T},\ldots,N_{G}/N_{T})^{\prime}$. If weights are known and $w=\hat{w}$, the asymptotic variance of $\sqrt{N}\left(\hat{w}-w\right)$ is $\bar{\Sigma}_{w}=0$. On the other hand, if the observed sample is drawn from a population, then generally $\hat{w}\neq w$ and $\bar{\Sigma}_{w}\neq0$. To estimate $\bar{\Sigma}_{w}$ for the latter case, denote the diagonal matrix with elements of vector $\alpha$ on the main diagonal as $\diag(\alpha)$. Let 
\begin{equation}
\hat{\bar{\Sigma}}_{w}=N/N_{T}\left(\diag(\hat{w})-\hat{w}\hat{w}^{\prime}\right)\label{eq:est_Sigma_w}
\end{equation}
 be the estimator for $\bar{\Sigma}_{w}$. As shown in the proof of Lemma \ref{lem:wtauestimation}, $\hat{w}_{g}$ is the OLS estimator for $w_{g}$ in the regression $D^{(g)}_{i}=w_{g}T_{i}+v_{i}$ and hence a function of $\mathbf{D}=(D_{1},\ldots,D_{N})^{\prime}$, and $\hat{\bar{\Sigma}}_{w}/N$ is the heteroscedasticity robust covariance matrix of $\hat{w}$. The following lemma states that Assumption \ref{assu:delta} holds for the OLS estimators $\hat{w}$ and $\hat{\tau}$.
\end{rem}
\begin{lem}
\label{lem:wtauestimation}Assume that: (i) The data $\left\{ \left(Y_{i},X_{i},D_{i}\right),i=1,\ldots,N\right\} $ is an i.i.d. sample drawn from the population. (ii) Treatment probability is $E(\mathbf{1}^{\prime}D_{i}=1)=p_{T}$, where $0<p_{T}<1$. (iii) $E\left(\tilde{X}_{i}\tilde{X}^{\prime}_{i}\right)$ is finite and nonsingular. (iv) $E(\epsilon_{i}|\tilde{X}_{i})=0$ and $E(\epsilon^{2}_{i}\tilde{X}_{i}\tilde{X}^{\prime}_{i})=\Sigma_{x}$, where $\Sigma_{x}$ is finite and positive definite. Then Assumption \ref{assu:delta} in the main text holds, specifically: (a)
\[
\sqrt{N}\left(\begin{array}{c}
\hat{w}-w\\
\hat{\tau}-\tau
\end{array}\right)\xrightarrow{d}\mathcal{N}\left[0_{2G\times1},\left(\begin{array}{cc}
\bar{\Sigma}_{w} & 0\\
0 & \bar{\Sigma}_{\tau}
\end{array}\right)\right],
\]
where $\bar{\Sigma}_{w}=p^{-1}_{T}\left(\diag(w)-ww^{\prime}\right)$ is positive semi-definite, $\bar{\Sigma}_{\tau}$ is the lower-right $G\times G$ submatrix of $\bar{\Sigma}_{\left(\beta,\tau\right)}=\left[E\left(\tilde{X}_{i}\tilde{X}^{\prime}_{i}\right)\right]^{-1}\Sigma_{x}\left[E\left(\tilde{X}_{i}\tilde{X}^{\prime}_{i}\right)\right]^{-1}$ and is positive definite.

(b) Let $\hat{\bar{\Sigma}}_{\tau}$ be the lower-right $G\times G$ submatrix of $\hat{\bar{\Sigma}}_{(\beta,\tau)}$ in (\ref{eq:est_Sigma_btau}), and $\hat{\bar{\Sigma}}_{w}$ be defined in (\ref{eq:est_Sigma_w}), then $\hat{\bar{\Sigma}}_{w}-\bar{\Sigma}_{w}\xrightarrow{p}0$ and $\hat{\bar{\Sigma}}_{\tau}-\bar{\Sigma}_{\tau}\xrightarrow{p}0$ as $N\rightarrow\infty$.
\end{lem}
\begin{rem}
In the special case when $G=1$, $\hat{w}=w=1$ and $\hat{\bar{\Sigma}}_{w}=\bar{\Sigma}_{w}=0$, which satisfies Assumption \ref{assu:delta} trivially. When $w$ is known, the lemma simplifies to $\sqrt{N}(\hat{\tau}-\tau)\xrightarrow{d}\mathcal{N}(0_{G\times1},\bar{\Sigma}_{\tau})$ and $\hat{\bar{\Sigma}}_{\tau}-\bar{\Sigma}_{\tau}\xrightarrow{p}0$ for the regression in (\ref{eq:ols}).

Proof of Lemma \ref{lem:wtauestimation} is in Section \ref{subsec:proof_wtauestimation}. Combining Lemma \ref{lem:wtauestimation} with Theorem \ref{thm:rhob} in the main text, under the conditions in Lemma \ref{lem:wtauestimation}, $\sqrt{N}(\hat{\rho}_{b}-\rho_{b})\xrightarrow{d}\mathcal{N}(0,\bar{\sigma}^{2}_{b})$ with
\begin{align*}
\bar{\sigma}^{2}_{b} & =\exp(\tau)^{\prime}\bar{\Sigma}_{w}\exp(\tau)+\left(w\odot\exp(\tau)\right)^{\prime}\bar{\Sigma}_{\tau}\left(w\odot\exp(\tau)\right).
\end{align*}
\end{rem}

\subsection{Example 2: Staggered Difference-in-differences Design\label{subsec:Example_2_Staggered}}

This example considers a staggered difference-in-differences design where multiple groups start receiving treatment at different times. I focus on the case with a never-treated group and no treatment exit and use panel data for illustration.

Let $c(i)$ denote the period when individual $i$ first receives treatment, with $c(i)=\infty$ for never-treated individuals. Cohort $c$ is defined as the group of individuals with $c(i)=c$, and the event time $r_{it}=t-c(i)$ is time relative to treatment. Each $c,r$ combination is viewed as a subgroup, with $\tau(c,r)$ denoting the ATE on the log-transformed outcome for cohort $c$ at event time $r$. Treatment effects may be heterogeneous across $c$ and $r$, i.e., $\tau(c,r)\neq\tau(c^{\prime},r^{\prime})$ if $c\neq c^{\prime}$ or $r\neq r^{\prime}$. Heterogeneity-robust estimators for the model typically involve estimating $\tau(c,r)$ and the corresponding weights $w(c,r)$, then computing the ATE in log points $\bar{\tau}=\sum_{c}\sum_{r}w(c,r)\tau(c,r)$ for a target population $P$. As discussed in Section \ref{sec:ATE-in-Percentage} in the main text, $\bar{\tau}$ and $\exp(\bar{\tau})-1$ differ from the ATE in percentage points, $\bar{\rho}$. However, the estimators of $w$ and $\tau$ in these studies usually satisfy Assumption \ref{assu:delta} in the main text, and therefore form the basis for the estimation and inference methods of $\rho_{b}$ in this paper.

Below I discuss estimators of $\tau(c,r)$ and $w(c,r)$ in \citet{sun_estimating_2021}, which are adaptations of the OLS estimators in Example 1 to staggered difference-in-differences settings, if we view each $(c,r)$ combination as a sub-treatment group. For the estimation of $\tau(c,r)$, consider the model
\begin{equation}
\ln(Y_{it})=\alpha_{i}+\beta_{t}+X^{\prime}_{it}\gamma+\sum_{c\neq\infty}\sum_{r\neq-1}D_{it}(c,r)\tau(c,r)+\epsilon_{it},\label{eq:stagdid}
\end{equation}
where $Y_{it}$ is the outcome of individual $i$ at time $t$, $\alpha_{i}$, $\beta_{t}$ and $X_{it}$ are individual fixed effects, time fixed effects and control variables respectively. The dummy variable $D_{it}(c,r)=1$ if individual $i$ belongs to cohort $c$ and $r_{it}=r$. The never-treated group ($c=\infty$) serves as the control group, and the period immediately preceding treatment ($r=-1$) serves as the base period. Under assumptions of conditional parallel trends and no anticipation, OLS estimation of Eq. (\ref{eq:stagdid}) yields estimators of $\tau(c,r)$ and the covariance matrix that satisfy Assumption \ref{assu:delta}, see Propositions 5 and 6 in \citet{sun_estimating_2021}.\footnote{With time-varying covariates, additional assumptions are needed. See Footnote 5 in the main text for a discussion.}

The estimation of $w(c,r)$ depends on the type of ATE of interest. Define $\mathcal{A}\subset\{(c,r):c\neq\infty,r\neq-1\}$ as a subset of the sub-treatment groups ($c,r$ combinations) for which to calculate the ATE in percentage points. This choice of $\mathcal{A}$ determines the target population $P$: $P$ is the population of observations in the selected $(c,r)$ cells. Following \citet{callaway_differenceindifferences_2021}, I define $\mathcal{A}=\{(c,r):c\neq\infty,r=r^{\ast}\}$ for ATE of event time $r^{\ast}$, with $r^{\ast}$ possibly negative to allow for event studies. Also, $\mathcal{A}=\{(c,r):c=c^{\ast},r\geqslant0\}$ for ATT of cohort $c^{\ast}$, $\mathcal{A}=\{(c,r):c+r=t,r\geqslant0\}$ for ATT of calendar time $t$, and $\mathcal{A}=\{(c,r):c\neq\infty,r\geqslant0\}$ for ATT of all treated units. Let $w^{\mathcal{A}}(c,r)$ be the share of subgroup $(c,r)$ in $\mathcal{A}$, with $w^{\mathcal{A}}(c,r)=0$ if $(c,r)\notin\mathcal{A}$. Then $w^{\mathcal{A}}(c^{\ast},r^{\ast})=E\left[(c,r)=(c^{\ast},r^{\ast})|(c,r)\in\mathcal{A}\right]$.

Treating each $(c,r)$ combination in $\mathcal{A}$ as a sub-treatment group $g$ and estimating $w^{\mathcal{A}}(c,r)$ as in Example 1: 
\begin{equation}
\hat{w}^{\mathcal{A}}(c,r)=\sum_{i}\sum_{t}D_{it}(c,r)/\left(\sum_{(c^{\ast},r^{\ast})\in\mathcal{A}}\sum_{i}\sum_{t}D_{it}(c^{\ast},r^{\ast})\right),\label{eq:wA}
\end{equation}
where the numerator is the sample size of cohort $c$ at event time $r$, and the denominator is the sample size of all units in set $\mathcal{A}$. Equivalently, $\hat{w}^{\mathcal{A}}(c,r)$ is the OLS estimator in the regression $D_{it}(c,r)=w^{\mathcal{A}}(c,r)T^{\mathcal{A}}_{it}+v_{it},$ where $T^{\mathcal{A}}_{it}=\sum_{(c,r)\in\mathcal{A}}D_{it}(c,r)$ is the indicator that individual $i$ in period $t$ belongs to set $\mathcal{A}$. Let $\hat{w}^{\mathcal{A}}$ be the vector of $\hat{w}^{\mathcal{A}}(c,r)$ for all $(c,r)$ combinations in $\mathcal{A}$. If the sample in $\mathcal{A}$ is randomly drawn from the population, the asymptotic covariance matrix of $\sqrt{N}\left(\hat{w}^{\mathcal{A}}-w^{\mathcal{A}}\right)$ is estimated as $\hat{\bar{\Sigma}}^{\mathcal{A}}_{w}=N/N_{\mathcal{A}}\left(\diag(\hat{w}^{\mathcal{A}})-\hat{w}^{\mathcal{A}}\hat{w}^{\mathcal{A}\prime}\right)$, where $N_{\mathcal{A}}$ is the size of $\mathcal{A}$. The estimators $\hat{w}^{\mathcal{A}}$ and $\hat{\bar{\Sigma}}^{\mathcal{A}}_{w}$ satisfy Assumption \ref{assu:delta} by Lemma \ref{lem:wtauestimation} and are identical to those in \citet{sun_estimating_2021}. When $\mathcal{A}=\{(c,r):c=c^{\ast},r\geqslant0\}$, i.e., when we are interested in the ATT of cohort $c^{\ast}$, then $\bar{\Sigma}^{\mathcal{A}}_{w}=0$ as weights for each period are equal and known.

\subsection{Proof of Lemma \ref{lem:wtauestimation} \label{subsec:proof_wtauestimation}}
\begin{proof}
For the estimator $\hat{\tau}$, observe that
\[
\left(\begin{array}{c}
\hat{\beta}\\
\hat{\tau}
\end{array}\right)-\left(\begin{array}{c}
\beta\\
\tau
\end{array}\right)=(\tilde{\mathbf{X}}^{\prime}\tilde{\mathbf{X}})^{-1}\tilde{\mathbf{X}}^{\prime}\epsilon=\left(\frac{1}{N}\sum^{N}_{i=1}\tilde{X}_{i}\tilde{X}^{\prime}_{i}\right)^{-1}\left(\frac{1}{N}\sum^{N}_{i=1}\tilde{X}_{i}\epsilon_{i}\right).
\]
By the law of large numbers, $N^{-1}\sum^{N}_{i=1}\tilde{X}_{i}\tilde{X}^{\prime}_{i}\xrightarrow{p}E(\tilde{X}_{i}\tilde{X}^{\prime}_{i})$, which is a finite nonsingular matrix. Denote the bottom $G$ rows of $(N^{-1}\sum^{N}_{i=1}\tilde{X}_{i}\tilde{X}^{\prime}_{i})^{-1}$ and $\left[E(\tilde{X}_{i}\tilde{X}^{\prime}_{i})\right]^{-1}$ as $C_{D}$ and $\bar{C}_{D}$ respectively. Consequently, $C_{D}\xrightarrow{p}\bar{C}_{D}$, where $\bar{C}_{D}$ is a finite constant matrix. As a result, 
\begin{align}
\sqrt{N}\left(\hat{\tau}-\tau\right) & =C_{D}\frac{1}{\sqrt{N}}\sum^{N}_{i=1}\tilde{X}_{i}\epsilon_{i}=\left(C_{D}-\bar{C}_{D}\right)\frac{1}{\sqrt{N}}\sum^{N}_{i=1}\tilde{X}_{i}\epsilon_{i}+\frac{1}{\sqrt{N}}\sum^{N}_{i=1}\bar{C}_{D}\tilde{X}_{i}\epsilon_{i}.\label{eq:tauhat_sum-1}
\end{align}

For the estimator $\hat{w}$, $w_{g}=E(D^{(g)}_{i}=1|T_{i}=1)$ leads to $E(D^{(g)}_{i}|T_{i})=w_{g}T_{i}$. Define $v^{(g)}_{i}=D^{(g)}_{i}-E(D^{(g)}_{i}|T_{i})$ so that $E(v^{(g)}_{i}|T_{i})=0$ by construction. This yields: 
\begin{equation}
D^{(g)}_{i}=w_{g}T_{i}+v^{(g)}_{i},g=1,\ldots,G.\label{eq:ols_w-1}
\end{equation}
The OLS estimator of $w_{g}$ is
\begin{equation}
\hat{w}_{g}=\left(\sum^{N}_{i=1}T^{2}_{i}\right)^{-1}\left(\sum^{N}_{i=1}T_{i}D^{(g)}_{i}\right)=N_{g}/N_{T},\label{eq:what-1}
\end{equation}
recalling that $D^{(g)}_{i}$ are mutually exclusive dummies that sum up to $T_{i}$ so that $T^{2}_{i}=T_{i}$ and $T_{i}D^{(g)}_{i}=D^{(g)}_{i}$.

Plugging (\ref{eq:ols_w-1}) into (\ref{eq:what-1}), we obtain
\[
\hat{w}_{g}-w_{g}=\frac{1}{N_{T}}\sum^{N}_{i=1}T_{i}v^{(g)}_{i},g=1,\ldots,G,
\]
and hence with $v_{i}=(v^{(1)}_{i},\ldots,v^{(G)}_{i})^{\prime}$, we have $\hat{w}-w=N^{-1}_{T}\sum^{N}_{i=1}T_{i}v_{i}$ and 
\begin{align}
\sqrt{N}(\hat{w}-w) & =\left(\frac{N}{N_{T}}-\frac{1}{p_{T}}+\frac{1}{p_{T}}\right)\frac{1}{\sqrt{N}}\sum^{N}_{i=1}T_{i}v_{i}\nonumber \\
 & =\left(\frac{N}{N_{T}}-\frac{1}{p_{T}}\right)\frac{1}{\sqrt{N}}\sum^{N}_{i=1}T_{i}v_{i}+\frac{1}{\sqrt{N}}\sum^{N}_{i=1}\frac{1}{p_{T}}T_{i}v_{i}.\label{eq:what_sum-1}
\end{align}
Equation (\ref{eq:ols_w-1}) implies $D_{i}=wT_{i}+v_{i}$, hence $T_{i}v_{i}=T_{i}(D_{i}-wT_{i})=D_{i}-wT_{i}$ as $T^{2}_{i}=T_{i}$ and $T_{i}D_{i}=D_{i}$. Note also that $D_{i}D^{\prime}_{i}=\diag(D_{i})$, $E(T_{i})=p_{T}$, and $E(D_{i})=p_{T}w$, thus
\begin{align}
E\left[(T_{i}v_{i})(T_{i}v_{i})^{\prime}\right] & =E\left[\left(D_{i}-wT_{i}\right)\left(D_{i}-wT_{i}\right)^{\prime}\right]\nonumber \\
 & =E\left[\diag(D_{i})-wD^{\prime}_{i}-D_{i}w^{\prime}+T_{i}ww^{\prime}\right]\nonumber \\
 & =p_{T}\diag(w)-p_{T}ww^{\prime}-p_{T}ww^{\prime}+p_{T}ww^{\prime}\nonumber \\
 & =p_{T}\left(\diag(w)-ww^{\prime}\right).\label{eq:cov_sv-1}
\end{align}

Combining (\ref{eq:tauhat_sum-1}) and (\ref{eq:what_sum-1}) yields 
\begin{equation}
\sqrt{N}\left(\begin{array}{c}
\hat{w}-w\\
\hat{\tau}-\tau
\end{array}\right)=\frac{1}{\sqrt{N}}\sum^{N}_{i=1}\left(\begin{array}{c}
p^{-1}_{T}T_{i}v_{i}\\
\bar{C}_{D}\tilde{X}_{i}\epsilon_{i}
\end{array}\right)+\left(\begin{array}{c}
\left(\frac{N}{N_{T}}-\frac{1}{p_{T}}\right)\frac{1}{\sqrt{N}}\sum^{N}_{i=1}T_{i}v_{i}\\
\left(C_{D}-\bar{C}_{D}\right)\frac{1}{\sqrt{N}}\sum^{N}_{i=1}\tilde{X}_{i}\epsilon_{i}
\end{array}\right).\label{eq:joint-1}
\end{equation}
For the first term on the right-hand side, first note that $E\left(\bar{C}_{D}\tilde{X}_{i}\epsilon_{i}\right)=0$ and $E\left(p^{-1}_{T}T_{i}v_{i}\right)=0$ as $E(\epsilon_{i}|\tilde{X}_{i})=0$ and $E(v_{i}|T_{i})=0$. Next, $E(\epsilon^{2}_{i}\tilde{X}_{i}\tilde{X}^{\prime}_{i})=\Sigma_{x}$ by assumption, $E(T^{2}_{i}v_{i}v^{\prime}_{i})=\left(\diag(w)-ww^{\prime}\right)p_{T}$ from (\ref{eq:cov_sv-1}), and $E\left(\tilde{X}_{i}\epsilon_{i}T_{i}v^{\prime}_{i}\right)=E_{\tilde{X}}\left[E\left(\tilde{X}_{i}\epsilon_{i}T_{i}v^{\prime}_{i}|\tilde{X}_{i}\right)\right]=E_{\tilde{X}}\left[\tilde{X}_{i}E(\epsilon_{i}|\tilde{X}_{i})T_{i}v^{\prime}_{i}\right]=0$ as $T_{i}=\mathbf{1}^{\prime}_{G}D_{i}$ and $v_{i}=D_{i}-wT_{i}$ are both functions of $D_{i}$ and hence of $\tilde{X}_{i}=[X^{\prime}_{i},D^{\prime}_{i}]^{\prime}$. By the central limit theorem, 
\begin{equation}
\frac{1}{\sqrt{N}}\sum^{N}_{i=1}\left(\begin{array}{c}
p^{-1}_{T}T_{i}v_{i}\\
\bar{C}_{D}\tilde{X}_{i}\epsilon_{i}
\end{array}\right)\xrightarrow{d}\mathcal{N}\left[0_{2G\times1},\left(\begin{array}{cc}
\bar{\Sigma}_{w} & 0\\
0 & \bar{\Sigma}_{\tau}
\end{array}\right)\right],\label{eq:conv_d-1}
\end{equation}
where $\bar{\Sigma}_{w}=p^{-1}_{T}\left[\diag(w)-ww^{\prime}\right]$, and $\bar{\Sigma}_{\tau}=\bar{C}_{D}\Sigma_{x}\bar{C}^{\prime}_{D}$. Since $\bar{C}_{D}$ is the bottom $G$ rows of $\left[E(\tilde{X}_{i}\tilde{X}^{\prime}_{i})\right]^{-1}$, $\bar{\Sigma}_{\tau}=\bar{C}_{D}\Sigma_{x}\bar{C}^{\prime}_{D}$ is the lower-right $G\times G$ submatrix of $\bar{\Sigma}_{(\beta,\tau)}=\left[E(\tilde{X}_{i}\tilde{X}^{\prime}_{i})\right]^{-1}\Sigma_{x}\left[E(\tilde{X}_{i}\tilde{X}^{\prime}_{i})\right]^{-1}$. Note that $\bar{\Sigma}_{\left(\beta,\tau\right)}$ is positive definite as $E(\tilde{X}_{i}\tilde{X}^{\prime}_{i})$ is finite and nonsingular, and $\Sigma_{x}$ is positive definite. Thus $\bar{\Sigma}_{\tau}$ is positive definite, as for any nonzero $G\times1$ vector $\alpha$, let $\tilde{\alpha}=(0^{\prime}_{k_{x}\times1},\alpha^{\prime})^{\prime}$, then $\tilde{\alpha}^{\prime}\bar{\Sigma}_{(\beta,\tau)}\tilde{\alpha}=\alpha^{\prime}\bar{\Sigma}_{\tau}\alpha>0$.

For the second term on the R.H.S of (\ref{eq:joint-1}), first note that following (\ref{eq:conv_d-1}), $N^{-1/2}\sum^{N}_{i=1}\tilde{X}_{i}\epsilon_{i}\xrightarrow{d}\mathcal{N}(0,\Sigma_{x})$, $N^{-1/2}\sum^{N}_{i=1}T_{i}v_{i}\xrightarrow{d}\mathcal{N}\left(0,p_{T}\left(\diag(w)-ww^{\prime}\right)\right)$. Next, since $T_{i}=D^{\prime}_{i}\mathbf{1}_{G}$ is i.i.d. across $i$ and $E(T_{i})=p_{T}$, by the law of large numbers $N_{T}/N=\sum^{N}_{i=1}T_{i}/N\xrightarrow{p}p_{T}>0$, hence $\left(N/N_{T}-1/p_{T}\right)\xrightarrow{p}0$. Also, $C_{D}-\bar{C}_{D}\xrightarrow{p}0$. By Slutsky's theorem, the second term on the right-hand side of (\ref{eq:joint-1}) converges to 0 in distribution and hence in probability. Taken together, 
\begin{equation}
\sqrt{N}\left(\begin{array}{c}
\hat{w}-w\\
\hat{\tau}-\tau
\end{array}\right)\xrightarrow{d}\mathcal{N}\left[0_{2G\times1},\left(\begin{array}{cc}
\bar{\Sigma}_{w} & 0\\
0 & \bar{\Sigma}_{\tau}
\end{array}\right)\right].\label{eq:joit_wtau-1}
\end{equation}

With $0<p_{T}<1$ and $0<w_{g}\leqslant1$ for each included group, the matrix $\bar{\Sigma}_{w}=p^{-1}_{T}\left[\diag(w)-ww^{\prime}\right]$ is finite. It remains to show $\diag(w)-ww^{\prime}$ is positive semi-definite. Let $\alpha=(a_{1},\ldots,a_{G})^{\prime}$ be a nonzero $G\times1$ vector, and $\bar{a}=w^{\prime}\alpha$. We have
\begin{align*}
\alpha^{\prime}\left(\diag(w)-ww^{\prime}\right)\alpha & =\sum^{G}_{g=1}w_{g}a^{2}_{g}-\bar{a}^{2}=\sum^{G}_{g=1}w_{g}(a_{g}-\bar{a}+\bar{a})^{2}-\bar{a}^{2}\\
 & =\sum^{G}_{g=1}w_{g}(a_{g}-\bar{a})^{2}\geqslant0.
\end{align*}
Hence $\bar{\Sigma}_{w}$ is positive semi-definite.

For part (b) of the lemma, it is easy to see that $\hat{\bar{\Sigma}}_{(\beta,\tau)}-\bar{\Sigma}_{(\beta,\tau)}\xrightarrow{p}0$ following the standard theory for heteroscedasticity-robust standard errors, see e.g., Proposition 2.4 in \citet{hayashi_econometrics_2000}. Since $\hat{\bar{\Sigma}}_{\tau}$ and $\bar{\Sigma}_{\tau}$ are the lower-right $G\times G$ submatrices of $\hat{\bar{\Sigma}}_{(\beta,\tau)}$ and $\bar{\Sigma}_{(\beta,\tau)}$ respectively, we have $\hat{\bar{\Sigma}}_{\tau}\xrightarrow{p}\bar{\Sigma}_{\tau}.$ Furthermore, since $\hat{w}\xrightarrow{p}w$ and $N_{T}/N\xrightarrow{p}p_{T}>0$, $\hat{\bar{\Sigma}}_{w}=N/N_{T}\left(\diag(\hat{w})-\hat{w}\hat{w}^{\prime}\right)\xrightarrow{p}\bar{\Sigma}_{w}=p^{-1}_{T}\left(\diag(w)-ww^{\prime}\right)$.

To see that $\hat{\bar{\Sigma}}_{w}$ is the heteroscedasticity robust estimator of $\bar{\Sigma}_{w}$, observe that $\sqrt{N}(\hat{w}-w)=N^{1/2}N^{-1}_{T}\sum^{N}_{i=1}T_{i}v_{i}$, so the heteroscedasticity robust estimator of $\bar{\Sigma}_{w}$ is $N/N^{2}_{T}\sum^{N}_{i=1}T^{2}_{i}\hat{v}_{i}\hat{v}^{\prime}_{i},$ where $\hat{v}_{i}=D_{i}-\hat{w}T_{i}$ is the OLS residual. Following analogous argument for (\ref{eq:cov_sv-1}), 
\begin{align*}
T^{2}_{i}\hat{v}_{i}\hat{v}^{\prime}_{i} & =T^{2}_{i}\left(D_{i}-\hat{w}T_{i}\right)\left(D_{i}-\hat{w}T_{i}\right)^{\prime}\\
 & =\diag(D_{i})-\hat{w}D^{\prime}_{i}-D_{i}\hat{w}^{\prime}+T_{i}\hat{w}\hat{w}^{\prime}.
\end{align*}
Since $N^{-1}_{T}\sum^{N}_{i=1}D_{i}=\hat{w}$, and $\sum^{N}_{i=1}T_{i}=N_{T}$, the heteroscedasticity robust variance estimator is
\begin{align*}
\frac{N}{N^{2}_{T}}\sum^{N}_{i=1}T^{2}_{i}\hat{v}_{i}\hat{v}^{\prime}_{i} & =\frac{N}{N_{T}}\left[\frac{1}{N_{T}}\sum^{N}_{i=1}\left(\diag(D_{i})-\hat{w}D^{\prime}_{i}-D_{i}\hat{w}^{\prime}+T_{i}\hat{w}\hat{w}^{\prime}\right)\right]\\
 & =\frac{N}{N_{T}}\left(\diag(\hat{w})-\hat{w}\hat{w}^{\prime}-\hat{w}\hat{w}^{\prime}+\hat{w}\hat{w}^{\prime}\right)\\
 & =\frac{N}{N_{T}}\left(\diag(\hat{w})-\hat{w}\hat{w}^{\prime}\right),
\end{align*}
which is exactly $\hat{\bar{\Sigma}}_{w}$.
\end{proof}

\section{\label{sec:normal heteogeneity}Estimation and Inference Under Normally Distributed Heterogeneity}

This section develops estimation and inference procedures for $\bar{\rho}$ allowing for treatment effect heterogeneity both across and within subgroups. It includes two special cases: (i) $G=1$, representing pure idiosyncratic heterogeneity, and (ii) constant effects within groups, in which case, the estimator reduces to $\hat{\rho}_{b}$.

As discussed before, point identification of $\bar{\rho}$ is generally impossible without further distributional assumptions. The estimation and inference methods here rely crucially on the assumption that the log-transformed potential outcomes follow a joint normal distribution within each subgroup. The rationale behind this assumption is that a key motivation for log-transforming the outcome variable is to reduce skewness and achieve a more normal distribution. The joint normality assumption on the potential outcomes may be too restrictive in practice, in which case the estimation and inference methods for $\rho_{b}$ discussed in the main text provide an easy-to-calculate lower-bound approach.

I impose the following distributional assumption.
\begin{assumption}
\label{assu:normal}For $g=1,\ldots,G$, 
\[
\left.\left(\begin{array}{c}
\ln Y_{1i}\\
\ln Y_{0i}
\end{array}\right)\right|D^{(g)}_{i}=1\sim\mathcal{N}\left(\left(\begin{array}{c}
\mu_{1g}\\
\mu_{0g}
\end{array}\right),\left(\begin{array}{cc}
s^{2}_{1g} & r_{g}s_{1g}s_{0g}\\
r_{g}s_{1g}s_{0g} & s^{2}_{0g}
\end{array}\right)\right),
\]
with $\mu_{1g}-\mu_{0g}=\tau_{g}$, $0<s^{2}_{1g}<\infty$, $0<s^{2}_{0g}<\infty$.
\end{assumption}
The parameter $r_{g}$ represents the Pearson correlation coefficient between $\ln Y_{1i}$ and $\ln Y_{0i}$ in subgroup $g$, satisfying $-1\leqslant r_{g}\leqslant1$. Under Assumption \ref{assu:normal}, the individual log-point effects follow $\tau_{i}|D^{(g)}_{i}=1\sim\mathcal{N}(\tau_{g},\varsigma^{2}_{g})$, where $\varsigma^{2}_{g}=s^{2}_{1g}+s^{2}_{0g}-2r_{g}s_{1g}s_{0g}$.

According to the mean of log-normal distributions, if $x\sim\mathcal{N}(\mu_{x},\sigma^{2}_{x})$, then $E\left[\exp(x)\right]=\exp(\mu_{x}+0.5\sigma^{2}_{x})$. Therefore, under Assumption \ref{assu:normal} the average proportional effect for subgroup $g$ is $\rho_{g}=E(\left.\exp(\tau_{i})\right|D^{(g)}_{i}=1)-1=\exp(\tau_{g}+0.5\varsigma^{2}_{g})-1$. The target ATE in percentage points becomes\footnote{I use the notation $\bar{\rho}^{+}$ here to emphasize the functional form of $\bar{\rho}$ in this setting.}
\begin{align}
\bar{\rho}^{+} & =\sum^{G}_{g=1}w_{g}\exp(\tau_{g}+\frac{1}{2}\varsigma^{2}_{g})-1\label{eq:rhoplus}\\
 & =\sum^{G}_{g=1}w_{g}\exp\left[\tau_{g}+\frac{1}{2}\left(s^{2}_{1g}+s^{2}_{0g}-2r_{g}s_{1g}s_{0g}\right)\right]-1.\nonumber 
\end{align}
When treatment effects are constant within groups, $r_{g}=1$ and $s^{2}_{1g}=s^{2}_{0g}$, hence $\varsigma^{2}_{g}=0$ and $\bar{\rho}^{+}$ reduces to $\rho_{b}$ in (\ref{eq:rho_b}) in the main text.

Estimators of $\bar{\rho}^{+}$ can be obtained by replacing each parameter with its estimator. The estimators for $w$ and $\tau$ in semi-log regressions have already been discussed in Section \ref{sec:examples}. The variances $s^{2}_{1g}$ and $s^{2}_{0g}$ can usually be estimated with observed values of $\ln Y_{i}$ in the treated sample and untreated sample, as discussed in Example 1 (continued) below. The correlation coefficient $r_{g}$ is not point identified from the marginal distribution of $\ln Y_{i}$ without further assumptions, as we can only observe one of $\ln Y_{1i}$ and $\ln Y_{0i}$ for the same unit \citep{heckman_making_1997}. However, $-1\leqslant r_{g}\leqslant1$ and narrower bounds or specific values of $r_{g}$ can be imposed based on application-specific assumptions. For example, the rank preservation assumption, widely used in treatment effect heterogeneity analysis \citep{heckman_making_1997}, posits perfect rank correlation between $\ln Y_{1i}$ and $\ln Y_{0i}$. Under bivariate normality, this implies $r_{g}=1$.\footnote{For bivariate normal distribution, Kendall's rank correlation $r_{k}=\frac{2}{\pi}\arcsin(r_{p})$ and Spearman's rank correlation $r_{s}=\frac{6}{\pi}\arcsin(\frac{r_{p}}{2})$, where $r_{p}$ denotes the Pearson correlation \citep{kruskal_ordinal_1958}. Thus $r_{k}=1$ or $r_{s}=1$ implies $r_{p}=1$.} Another common assumption is independence between treatment effects and baseline outcome \citep{heckman_making_1997,djebbari_heterogeneous_2008}. Under this restriction, $cov(\ln Y_{1i},\ln Y_{0i}|D^{(g)}_{i}=1)=cov(\ln Y_{0i}+\tau_{i},\ln Y_{0i}|D^{(g)}_{i}=1)=s^{2}_{0g}$, yielding $r_{g}=s_{0g}/s_{1g}$ and $\varsigma^{2}_{g}=s^{2}_{1g}-s^{2}_{0g}$. If $\ln Y_{0i}$ and $\tau_{i}$ are positively correlated, then $r_{g}>s_{0g}/s_{1g}$ and $\varsigma^{2}_{g}<s^{2}_{1g}-s^{2}_{0g}$. In many applications, non-negative correlation is plausible, implying $r_{g}\geqslant0$.

With $s_{1g}>0$ and $s_{0g}>0$, Eq. (\ref{eq:rhoplus}) indicates that $\bar{\rho}^{+}$ is strictly decreasing in each $r_{g}$. This monotonicity enables sensitivity analysis by evaluating $\bar{\rho}^{+}$ across different values of $\mathbf{r}=(r_{1},\ldots,r_{G})^{\prime}$. Below I develop the estimation and inference procedures for $\bar{\rho}^{+}$, treating $r_{g}$s as known constants. Let $\hat{s}^{2}_{1g}$ and $\hat{s}^{2}_{0g}$ be respective estimators of $s^{2}_{1g}$ and $s^{2}_{0g}$. The estimator of $\varsigma^{2}_{g}$ is
\begin{equation}
\hat{\varsigma}^{2}_{g}=\hat{s}^{2}_{1g}+\hat{s}^{2}_{0g}-2r_{g}\hat{s}_{1g}\hat{s}_{0g},\label{eq:varsigma}
\end{equation}
and the corresponding estimator for $\bar{\rho}^{+}$ is 
\begin{equation}
\hat{\rho}^{+}_{b}=\sum^{G}_{g=1}\hat{w}_{g}\exp(\hat{\tau}_{g}+0.5\hat{\varsigma}^{2}_{g})-1.\label{eq:rho_b_plus}
\end{equation}
Define $\gamma_{s1}=(s^{2}_{11},\ldots,s^{2}_{1G})^{\prime}$ and $\gamma_{s0}=(s^{2}_{01},\ldots,s^{2}_{0G})^{\prime}$ with estimators $\hat{\gamma}_{s1}=(\hat{s}^{2}_{11},\ldots,\hat{s}^{2}_{1G})^{\prime}$ and $\hat{\gamma}_{s0}=(\hat{s}^{2}_{01},\ldots,\hat{s}^{2}_{0G})^{\prime}$. The asymptotic properties of $\hat{\rho}^{+}_{b}$ follow from an extended version of Assumption \ref{assu:delta} in the main text.
\begin{assumption}
\label{assu:deltaplus} Let $\delta^{+}=(w^{\prime},\tau^{\prime},\gamma^{\prime}_{s1},\gamma^{\prime}_{s0})^{\prime}$ and $\hat{\delta}^{+}=(\hat{w}^{\prime},\hat{\tau}^{\prime},\hat{\gamma}^{\prime}_{s1},\hat{\gamma}^{\prime}_{s0})^{\prime}$. There exists a finite matrix $\bar{\Sigma}_{\delta^{+}}$ and its estimator $\hat{\bar{\Sigma}}_{\delta^{+}}$ such that $\sqrt{N}(\hat{\delta}^{+}-\delta^{+})\xrightarrow{d}\mathcal{N}\left(0,\bar{\Sigma}_{\delta^{+}}\right)$ and $\hat{\bar{\Sigma}}_{\delta^{+}}\xrightarrow{p}\bar{\Sigma}_{\delta^{+}}$.
\end{assumption}
Example 1 (continued) below demonstrates how this assumption is satisfied in semi-log regression models. Similar to Assumption \ref{assu:delta}, Assumption \ref{assu:deltaplus} accommodates the case of known weights, where the variance of $\hat{w}$ and its covariances with other estimators are all zero.

Assumption \ref{assu:deltaplus} indicates that $\hat{w}_{g}$, $\hat{\tau}_{g},$ $\hat{s}^{2}_{1g}$ and $\hat{s}^{2}_{0g}$ are consistent, hence $\hat{\rho}^{+}_{b}$ is consistent by the continuous mapping theorem. The asymptotic distribution follows from the delta method. Observe that $\partial\bar{\rho}^{+}/\partial w_{g}=\exp(\tau_{g}+0.5\varsigma^{2}_{g})$, $\partial\bar{\rho}^{+}/\partial\tau_{g}=w_{g}\exp(\tau_{g}+0.5\varsigma^{2}_{g})$, $\partial\bar{\rho}^{+}/\partial s^{2}_{1g}=0.5w_{g}\exp(\tau_{g}+0.5\varsigma^{2}_{g})(1-r_{g}s_{0g}/s_{1g})$, and $\partial\bar{\rho}^{+}/\partial s^{2}_{0g}=0.5w_{g}\exp(\tau_{g}+0.5\varsigma^{2}_{g})(1-r_{g}s_{1g}/s_{0g})$. Consequently, the gradient vector is
\begin{equation}
\nabla_{\delta^{+}}\bar{\rho}^{+}=\left[\begin{array}{l}
\exp(\tau+0.5\gamma_{\varsigma})\\
w\odot\exp(\tau+0.5\gamma_{\varsigma})\\
\frac{1}{2}\diag^{G}_{g=1}\left\{ 1-\frac{r_{g}s_{0g}}{s_{1g}}\right\} w\odot\exp(\tau+0.5\gamma_{\varsigma})\\
\frac{1}{2}\diag^{G}_{g=1}\left\{ 1-\frac{r_{g}s_{1g}}{s_{0g}}\right\} w\odot\exp(\tau+0.5\gamma_{\varsigma})
\end{array}\right],\label{eq:nablarhoplus}
\end{equation}
where $\gamma_{\varsigma}=(\varsigma^{2}_{1},\ldots,\varsigma^{2}_{G})^{\prime}$ and $\odot$ denotes element-wise multiplication.
\begin{thm}
\label{thm:idio}Suppose that Assumptions \ref{assu:normal} and \ref{assu:deltaplus} hold. Then $E\left[\left(Y_{1i}-Y_{0i}\right)/Y_{0i}\right]=\bar{\rho}^{+}$, and for given values of $r_{1},\ldots,r_{G}$,
\[
\sqrt{N}(\hat{\rho}^{+}_{b}-\bar{\rho}^{+})\xrightarrow{d}\mathcal{N}\left(0,\bar{\sigma}^{+2}_{b}\right),
\]
where $\bar{\sigma}^{+2}_{b}=\left(\nabla_{\delta^{+}}\bar{\rho}^{+}\right)^{\prime}\bar{\Sigma}_{\delta^{+}}\left(\nabla_{\delta^{+}}\bar{\rho}^{+}\right)$. The consistent variance estimator is $\hat{\bar{\sigma}}^{+2}_{b}=\left(\widehat{\nabla_{\delta^{+}}\bar{\rho}^{+}}\right)^{\prime}\hat{\bar{\Sigma}}_{\delta^{+}}\left(\widehat{\nabla_{\delta^{+}}\bar{\rho}^{+}}\right)$ where $\widehat{\nabla_{\delta^{+}}\bar{\rho}^{+}}$ replaces parameters with their consistent estimators.
\end{thm}
This theorem follows from the continuous mapping theorem and the delta method. Proof is omitted. Theorem \ref{thm:idio} enables inference about $\bar{\rho}^{+}$ for any specified correlation structure. When the correlation coefficients are unknown, researchers can report results for different plausible values of $r_{g}$ and compute bounds.

Theorem \ref{thm:idio} readily adapts to various special cases encountered in practice, such as the corollary below.
\begin{cor}
\label{cor:equals0}Suppose $s^{2}_{0g}=s^{2}_{0}$ for all $g$ and $\gamma_{s0}=s^{2}_{0}$. Then Theorem \ref{thm:idio} holds with updated $\gamma_{s0}$ in Assumption \ref{assu:deltaplus} and the modified gradient
\begin{equation}
\nabla_{\delta^{+}}\bar{\rho}^{+}=\left[\begin{array}{l}
\exp(\tau+0.5\gamma_{\varsigma})\\
w\odot\exp(\tau+0.5\gamma_{\varsigma})\\
\frac{1}{2}\diag^{G}_{g=1}\left\{ 1-\frac{r_{g}s_{0g}}{s_{1g}}\right\} w\odot\exp(\tau+0.5\gamma_{\varsigma})\\
\frac{1}{2}\sum^{G}_{g=1}\left(1-\frac{r_{g}s_{1g}}{s_{0g}}\right)w_{g}\exp(\tau_{g}+0.5\varsigma^{2}_{g})
\end{array}\right].\label{eq:nablarhoplus-1}
\end{equation}
\end{cor}
Similarly, when variances of log potential outcomes follow a known parametric form or when certain groups share common variance parameters, the theorem applies with appropriate modifications to the parameter vector and gradient. The key requirement is that Assumption \ref{assu:deltaplus} holds for the specific parameterization chosen.

The example below shows how Assumption \ref{assu:deltaplus} is satisfied in semi-log regression models.

\subsubsection*{Example 1 (Continued): Semi-log Regression Model\label{subsec:OLS_continued}}

I now extend Example 1 in Section \ref{sec:examples} to accommodate normally distributed within-group heterogeneity while maintaining the focus on the average treatment effect on the treated.

Adapting Assumption \ref{assu:normal} to condition on $X_{i}$ with a linear specification, for individual $i$ in subgroup $g$:
\begin{equation}
\left.\left(\begin{array}{c}
\ln Y_{1i}\\
\ln Y_{0i}
\end{array}\right)\right|\left(D^{(g)}_{i}=1,X_{i}\right)\sim\mathcal{N}\left(\left(\begin{array}{c}
X^{\prime}_{i}\beta+\tau_{g}\\
X^{\prime}_{i}\beta
\end{array}\right),\left(\begin{array}{cc}
s^{2}_{1g} & r_{g}s_{1g}s_{0g}\\
r_{g}s_{1g}s_{0g} & s^{2}_{0g}
\end{array}\right)\right).\label{eq:olsass1}
\end{equation}
Since $\ln Y_{0i}$ is unobserved for the treatment group, identification of $s^{2}_{0g}$ requires additional structure. I assume that there are $G$ sub-control groups paired with the $G$ sub-treatment groups, such that
\begin{equation}
\left.\ln Y_{0i}\right|\left(C^{(g)}_{i}=1,X_{i}\right)\sim\mathcal{N}\left(X^{\prime}_{i}\beta,s^{2}_{0g}\right),\label{eq:olsass2}
\end{equation}
where $C^{(g)}_{i}$ is the dummy for the $g-$th sub-control group and is included in $X_{i}$. This assumption is plausible in settings with randomization within strata. For example, if individuals are randomly assigned to treatment within cities, the variance of $\ln Y_{0i}$ should be identical for treatment and control units within each city. Denote the sample size of sub-control group $g$ as $M_{g}$. Then $\sum^{G}_{g=1}M_{g}=N-N_{T}$ is the sample size of the whole control group. Let $\varpi_{g}=E(C^{(g)}_{i}=1|T_{i}=0)$ be the share of the sub-control group $g$ in the control population. Similar to $w_{g}$, $\varpi_{g}$ can be estimated by $\hat{\varpi}_{g}=M_{g}/(N-N_{T}).$

Equations (\ref{eq:olsass1}) and (\ref{eq:olsass2}) indicate that
\begin{align*}
\ln Y_{i} & =\ln Y_{0i}+T_{i}\tau_{i}=X^{\prime}_{i}\beta+D_{i}\tau+\epsilon_{i},
\end{align*}
where
\begin{align*}
\epsilon_{i} & =(1-T_{i})\left[\ln Y_{0i}-E(\ln Y_{0i}|X_{i},D_{i})\right]+T_{i}\left[\ln Y_{1i}-E(\ln Y_{1i}|X_{i},D_{i})\right].
\end{align*}
The regression specification above is the same as Eq. (\ref{eq:ols}), except for the structure of $\epsilon_{i}$. The estimators $\hat{w}$, $\hat{\tau}$ and $\hat{p}_{T}$ remain as defined in Section \ref{sec:examples}. For the variance parameters, observe that when $D^{(g)}_{i}=1$, $\epsilon_{i}=\ln Y_{1i}-E(\ln Y_{1i}|X_{i},D_{i})$ has variance $s^{2}_{1g}$, and when $C^{(g)}_{i}=1$, $\epsilon_{i}=\ln Y_{0i}-E(\ln Y_{0i}|X_{i},D_{i})$ has variance $s^{2}_{0g}$. This motivates the following estimators: $\hat{s}^{2}_{1g}=\sum^{N}_{i=1}\hat{\epsilon}^{2}_{i}D^{(g)}_{i}/N_{g},$ $\hat{s}^{2}_{0g}=\sum^{N}_{i=1}\hat{\epsilon}^{2}_{i}C^{(g)}_{i}/M_{g}$, where $\hat{\epsilon}_{i}$ is the OLS residual.
\begin{lem}
\label{lem:olsplus}Suppose conditions (i), (ii), and (iii) in Lemma \ref{lem:wtauestimation} hold. In addition, (i) $X_{i}$ now includes the sub-control group indicators $C_{i}=(C^{(1)}_{i},\ldots,C^{(G)}_{i})^{\prime}$ along with other covariates; (ii) for all $g$, $0<\varpi_{g}\leqslant1$; (iii) Condition (iv) in Lemma \ref{lem:wtauestimation} is updated to $\epsilon_{i}|\left(X_{i},D^{(g)}_{i}=1\right)\sim\mathcal{N}(0,s^{2}_{1g})$ and $\epsilon_{i}|\left(X_{i},C^{(g)}_{i}=1\right)\sim\mathcal{N}(0,s^{2}_{0g})$. Then Assumption \ref{assu:deltaplus} holds with: 
\[
\sqrt{N}\left(\hat{\delta}^{+}-\delta^{+}\right)\xrightarrow{d}\mathcal{N}\left(0,\bar{\Sigma}_{\delta^{+}}\right),
\]
where $\bar{\Sigma}_{\delta^{+}}=\diag\left\{ \bar{\Sigma}_{w},\bar{\Sigma}_{\tau},\bar{\Sigma}_{s1},\bar{\Sigma}_{s0}\right\} $, with $\bar{\Sigma}_{\tau}$ and $\bar{\Sigma}_{w}$ as defined in Lemma \ref{lem:wtauestimation}, $\bar{\Sigma}_{s1}=\diag^{G}_{g=1}\{2s^{4}_{1g}/(p_{T}w_{g})\}$ and $\bar{\Sigma}_{s0}=\diag^{G}_{g=1}\{2s^{4}_{0g}/[(1-p_{T})\varpi_{g}]\}$. A consistent estimator of $\bar{\Sigma}_{\delta^{+}}$ is $\hat{\bar{\Sigma}}_{\delta^{+}}=\diag\left\{ \hat{\bar{\Sigma}}_{w},\hat{\bar{\Sigma}}_{\tau},\hat{\bar{\Sigma}}_{s1},\hat{\bar{\Sigma}}_{s0}\right\} $, where $\hat{\bar{\Sigma}}_{\tau}$ and $\hat{\bar{\Sigma}}_{w}$ are as defined in Lemma \ref{lem:wtauestimation}, $\hat{\bar{\Sigma}}_{s1}=\diag^{G}_{g=1}\left\{ 2\hat{s}^{4}_{1g}/(\hat{p}_{T}\hat{w}_{g})\right\} $ and $\hat{\bar{\Sigma}}_{s0}=\diag^{G}_{g=1}\left\{ 2\hat{s}^{4}_{0,g}/[(1-\hat{p}_{T})\hat{\varpi}_{g}]\right\} $.
\end{lem}
Proof is in Section \ref{subsec:Proof-of-olsplus}.

Corresponding to the case in Corollary \ref{cor:equals0}, $s^{2}_{0g}$ is also identified if $s^{2}_{0g}=s^{2}_{0}$ for all $g$, and we have a single control group which also has $\Var(\ln Y_{0i})=s^{2}_{0}$. Then Eq. (\ref{eq:olsass2}) becomes
\begin{equation}
\left.\ln Y_{0i}\right|\left(T_{i}=0,X_{i}\right)\sim\mathcal{N}\left(X^{\prime}_{i}\beta,s^{2}_{0}\right).\label{eq:olsass2-1}
\end{equation}
The regression specification and estimators for $w$, $\tau$, $\gamma_{s1}$, $\bar{\Sigma}_{w}$, $\bar{\Sigma}_{\tau}$, $\bar{\Sigma}_{s1}$ remain unchanged. For $\gamma_{s0}=s^{2}_{0}$, observe that when $T_{i}=0$, $\epsilon_{i}=\ln Y_{0i}-E(\ln Y_{0i}|X_{i},D_{i})$ has variance $s^{2}_{0}$, hence the estimator of $s^{2}_{0}$ is $\hat{s}^{2}_{0}=\sum^{N}_{i=1}\hat{\epsilon}^{2}_{i}(1-T_{i})/(N-N_{T}).$
\begin{cor}
\label{cor:olsplus-1}Suppose all conditions in Lemma \ref{lem:wtauestimation} hold. In addition, $s^{2}_{0g}=s^{2}_{0}$ and condition (iv) in Lemma \ref{lem:wtauestimation} is updated to $\epsilon_{i}|\left(X_{i},D^{(g)}_{i}=1\right)\sim\mathcal{N}(0,s^{2}_{1g})$ and $\epsilon_{i}|\left(X_{i},T_{i}=0\right)\sim\mathcal{N}(0,s^{2}_{0})$. Then the results of Lemma \ref{lem:olsplus} hold with $\gamma_{s0}=s^{2}_{0}$, $\bar{\Sigma}_{s0}=2s^{4}_{0}/(1-p_{T})$, and $\hat{\bar{\Sigma}}_{s0}=2\hat{s}^{4}_{0}/(1-\hat{p}_{T})$.
\end{cor}
This framework extends naturally to staggered difference-in-differences designs by treating each cohort-event time combination $(c,r)$ as a sub-treatment group. For variance estimation in such settings, one can assume that $\Var(\ln Y_{0i})$ remains constant within each cohort across periods, so that $\Var(\ln Y_{0i}|c=c^{\ast},r=r^{\ast})=\Var(\ln Y_{0i}|c=c^{\ast},r\leqslant-1)$. This enables the use of pre-treatment observations to estimate $\Var(\ln Y_{0i})$ for each cohort.

\subsection{Proof of Lemma \ref{lem:olsplus}\label{subsec:Proof-of-olsplus}}

Since all conditions in Lemma \ref{lem:wtauestimation} still hold, the results for $\hat{w}$ and $\hat{\tau}$ in Lemma \ref{lem:wtauestimation} remain valid. For the variance estimators, observe that
\begin{align}
 & \sqrt{N}\left(\hat{s}^{2}_{1g}-s^{2}_{1g}\right)\nonumber \\
= & \frac{\sqrt{N}}{N_{g}}\sum^{N}_{i=1}D^{(g)}_{i}\left(\hat{\epsilon}^{2}_{i}-s^{2}_{1g}\right)=\frac{1}{\sqrt{N}}\frac{N}{N_{g}}\sum^{N}_{i=1}D^{(g)}_{i}\left[\left(\epsilon^{2}_{i}-s^{2}_{1g}\right)+\left(\hat{\epsilon}^{2}_{i}-\epsilon^{2}_{i}\right)\right]\nonumber \\
= & \left(\frac{1}{p_{T}w_{g}}+\frac{N}{N_{g}}-\frac{1}{p_{T}w_{g}}\right)\frac{1}{\sqrt{N}}\sum^{N}_{i=1}D^{(g)}_{i}\left(\epsilon^{2}_{i}-s^{2}_{1g}\right)+\frac{1}{\sqrt{N}}\frac{N}{N_{g}}\sum^{N}_{i=1}D^{(g)}_{i}\left(\hat{\epsilon}^{2}_{i}-\epsilon^{2}_{i}\right)\nonumber \\
= & \frac{1}{p_{T}w_{g}}\frac{1}{\sqrt{N}}\sum^{N}_{i=1}D^{(g)}_{i}\left(\epsilon^{2}_{i}-s^{2}_{1g}\right)+\left(\frac{N}{N_{g}}-\frac{1}{p_{T}w_{g}}\right)\frac{1}{\sqrt{N}}\sum^{N}_{i=1}D^{(g)}_{i}\left(\epsilon^{2}_{i}-s^{2}_{1g}\right)+\frac{1}{\sqrt{N}}\sum^{N}_{i=1}\frac{N}{N_{g}}D^{(g)}_{i}\left(\hat{\epsilon}^{2}_{i}-\epsilon^{2}_{i}\right).\label{eq:s2-1}
\end{align}
I first show that the second and third terms converge in probability to 0. For the second term, since $(D^{(g)}_{i},\epsilon_{i})$ are independently distributed across $i$ with $\epsilon_{i}|D^{(g)}_{i}=1\sim\mathcal{N}(0,s^{2}_{1g})$, we have $D^{(g)}_{i}\left(\epsilon^{2}_{i}-s^{2}_{1g}\right)$ independently distributed across $i$ with
\begin{equation}
E\left[D^{(g)}_{i}\left(\epsilon^{2}_{i}-s^{2}_{1g}\right)\right]=Pr(D^{(g)}_{i}=1)E(\epsilon^{2}_{i}-s^{2}_{1g}|D^{(g)}_{i}=1)=0,\label{eq:Ede2-1}
\end{equation}
and 
\begin{equation}
\Var\left[D^{(g)}_{i}\left(\epsilon^{2}_{i}-s^{2}_{1g}\right)\right]=Pr(D^{(g)}_{i}=1)E\left[\left(\epsilon^{2}_{i}-s^{2}_{1g}\right)^{2}|D^{(g)}_{i}=1\right]=2p_{T}w_{g}s^{4}_{1g},\label{eq:Ede22-1}
\end{equation}
where the last equality uses the property that for $\epsilon_{i}\sim\mathcal{N}(0,s^{2}_{1g})$, we have $E(\epsilon^{4}_{i})=3s^{4}_{1g}$. By the central limit theorem,
\begin{equation}
\frac{1}{\sqrt{N}}\sum^{N}_{i=1}D^{(g)}_{i}\left(\epsilon^{2}_{i}-s^{2}_{1g}\right)\xrightarrow{d}\mathcal{N}\left(0,2p_{T}w_{g}s^{4}_{1g}\right),g=1,\ldots,G.\label{eq:clt_de2-2}
\end{equation}
Since $N/N_{g}-1/(p_{T}w_{g})\xrightarrow{p}0$, the second term in (\ref{eq:s2-1}) converges to $0$ in probability by Slutsky's theorem.

For the third term, let $\epsilon=(\epsilon_{1},\ldots,\epsilon_{N})^{\prime}$ and $\hat{\epsilon}=M_{x}\epsilon$, where $M_{x}=I-P_{x}$ and $P_{x}=\tilde{\mathbf{X}}\left(\tilde{\mathbf{X}}^{\prime}\tilde{\mathbf{X}}\right)^{-1}\tilde{\mathbf{X}}^{\prime}$. Let $A_{g}=\diag^{N}_{i=1}\{D^{(g)}_{i}\}$, then
\begin{equation}
\frac{1}{\sqrt{N}}\sum^{N}_{i=1}\frac{N}{N_{g}}D^{(g)}_{i}\left(\hat{\epsilon}^{2}_{i}-\epsilon^{2}_{i}\right)=\frac{N}{N_{g}}\frac{1}{\sqrt{N}}\left(\hat{\epsilon}^{\prime}A_{g}\hat{\epsilon}-\epsilon^{\prime}A_{g}\epsilon\right)=\frac{N}{N_{g}}\frac{1}{\sqrt{N}}\epsilon^{\prime}B_{g}\epsilon,\label{eq:third-1}
\end{equation}
where $B_{g}=M_{x}A_{g}M_{x}-A_{g}=P_{x}A_{g}P_{x}-P_{x}A_{g}-A_{g}P_{x}$. Under the conditions specified in Lemma \ref{lem:olsplus}, $\Omega\equiv E\left(\epsilon\epsilon^{\prime}\mid\tilde{\mathbf{X}}\right)=\diag^{N}_{i=1}\left\{ \sum^{G}_{g=1}\left(D^{(g)}_{i}s^{2}_{1g}+C^{(g)}_{i}s^{2}_{0g}\right)\right\} $. We have $E(\epsilon^{\prime}B_{g}\epsilon|\tilde{\mathbf{X}})=\tr(B_{g}\Omega)$, and using Lemma A.1 in \citet{kelejian_specification_2010}, $\Var\left(\epsilon^{\prime}B_{g}\epsilon\mid\tilde{\mathbf{X}}\right)=2\tr(B_{g}\Omega B_{g}\Omega)$. To show $N^{-1/2}\epsilon^{\prime}B_{g}\epsilon\xrightarrow{p}0$, I verify that $B_{g}\Omega$ is $O_{p}(1/N)$ uniformly in $N$. Observe that 
\[
P_{x}A_{g}P_{x}\Omega=\frac{1}{N}\tilde{\mathbf{X}}\left(\frac{\tilde{\mathbf{X}}^{\prime}\tilde{\mathbf{X}}}{N}\right)^{-1}\frac{\tilde{\mathbf{X}}^{\prime}A_{g}\tilde{\mathbf{X}}}{N}\left(\frac{\tilde{\mathbf{X}}^{\prime}\tilde{\mathbf{X}}}{N}\right)^{-1}\tilde{\mathbf{X}}^{\prime}\Omega.
\]
Since $\tilde{X}_{i}$ are i.i.d. with $E\left(\tilde{X}_{i}\tilde{X}^{\prime}_{i}\right)$ finite and nonsingular, $(N^{-1}\tilde{\mathbf{X}}^{\prime}\tilde{\mathbf{X}})^{-1}\xrightarrow{p}\left[E(\tilde{X}_{i}\tilde{X}^{\prime}_{i})\right]^{-1}$ uniformly in $N$. Moreover, since $\tilde{X}_{i}\tilde{X}^{\prime}_{i}$ is i.i.d. with finite mean and $D^{(g)}_{i}\in\{0,1\}$ is included in $\tilde{X}_{i}$, we have that $D^{(g)}_{i}\tilde{X}_{i}\tilde{X}^{\prime}_{i}$ is i.i.d. and $O_{p}(1)$ uniformly, so $\tilde{\mathbf{X}}^{\prime}A_{g}\tilde{\mathbf{X}}/N=N^{-1}\sum^{N}_{i=1}\tilde{X}_{i}\tilde{X}^{\prime}_{i}D^{(g)}_{i}$ is bounded uniformly in $N$. Finally, elements of $\tilde{\mathbf{X}}^{\prime}\Omega$ are also $O_{p}(1)$ uniformly in $N$ as $\tilde{X}_{i}$ is i.i.d. and $\Omega$ is diagonal with finite elements. Therefore, $P_{x}A_{g}P_{x}\Omega$ is $O_{p}(1/N)$ uniformly in $N$. Similarly, $P_{x}A_{g}\Omega$ and $A_{g}P_{x}\Omega$ are both $O_{p}(1/N)$ uniformly. In all, $B_{g}\Omega$ is $O_{p}(1/N)$ uniformly in $N$. Consequently, $\tr(B_{g}\Omega)$ is $O_{p}(1)$, and $\tr(B_{g}\Omega B_{g}\Omega)=\tr\left[\left(B_{g}\Omega\right)\left(B_{g}\Omega\right)\right]=\sum^{N}_{i=1}\sum^{N}_{j=1}(B_{g}\Omega)_{ij}(B_{g}\Omega)_{ji}$ is $O_{p}(1)$. As a result, $E\left(N^{-1/2}\epsilon^{\prime}B_{g}\epsilon\right)\rightarrow0$ and $\Var\left(N^{-1/2}\epsilon^{\prime}B_{g}\epsilon\right)\rightarrow0$, implying $N^{-1/2}\epsilon^{\prime}B_{g}\epsilon\xrightarrow{p}0$. Since $N/N_{g}\xrightarrow{p}1/(p_{T}w_{g})$, Eq. (\ref{eq:third-1}) converges to $0$ in probability.

Consequently, 
\[
\sqrt{N}\left(\hat{s}^{2}_{1g}-s^{2}_{1g}\right)=\frac{1}{p_{T}w_{g}}\frac{1}{\sqrt{N}}\sum^{N}_{i=1}D^{(g)}_{i}\left(\epsilon^{2}_{i}-s^{2}_{1g}\right)+o_{p}(1).
\]
Similarly 
\[
\sqrt{N}\left(\hat{s}^{2}_{0g}-s^{2}_{0g}\right)=\frac{1}{(1-p_{T})\varpi_{g}}\frac{1}{\sqrt{N}}\sum^{N}_{i=1}C^{(g)}_{i}\left(\epsilon^{2}_{i}-s^{2}_{0g}\right)+o_{p}(1).
\]
with
\begin{equation}
\frac{1}{\sqrt{N}}\sum^{N}_{i=1}C^{(g)}_{i}\left(\epsilon^{2}_{i}-s^{2}_{0g}\right)\xrightarrow{d}\mathcal{N}\left(0,2(1-p_{T})\varpi_{g}s^{4}_{0g}\right),g=1,\ldots,G.\label{eq:clt_de2-1-1}
\end{equation}

Combining the results above with (\ref{eq:joint-1}) in the proof of Lemma \ref{lem:wtauestimation}, 
\[
\sqrt{N}(\hat{\delta}^{+}-\delta^{+})=\sqrt{N}\left(\begin{array}{c}
\hat{w}-w\\
\hat{\tau}-\tau\\
\hat{\gamma}_{s1}-\gamma_{s1}\\
\hat{\gamma}_{s0}-\gamma_{s0}
\end{array}\right)=\frac{1}{\sqrt{N}}\sum^{N}_{i=1}\left(\begin{array}{c}
p^{-1}_{T}T_{i}v_{i}\\
\bar{C}_{D}\tilde{X}_{i}\epsilon_{i}\\
\frac{1}{p_{T}w_{1}}D^{(1)}_{i}\left(\epsilon^{2}_{i}-s^{2}_{11}\right)\\
\vdots\\
\frac{1}{p_{T}w_{G}}D^{(G)}_{i}\left(\epsilon^{2}_{i}-s^{2}_{1G}\right)\\
\frac{1}{(1-p_{T})\varpi_{1}}C^{(1)}_{i}\left(\epsilon^{2}_{i}-s^{2}_{01}\right)\\
\vdots\\
\frac{1}{(1-p_{T})\varpi_{G}}C^{(G)}_{i}\left(\epsilon^{2}_{i}-s^{2}_{0G}\right)
\end{array}\right)+o_{p}(1).
\]
For the covariances between $\hat{\tau}$ and $\left(\hat{\gamma}^{\prime}_{s1},\hat{\gamma}^{\prime}_{s0}\right)^{\prime}$, observe that $\epsilon_{i}|\tilde{X}_{i}$ is normally distributed, thus $E(\epsilon^{3}_{i}|\tilde{X}_{i})=0$, and
\[
E\left(\bar{C}_{D}\tilde{X}_{i}\epsilon_{i}D^{(g)}_{i}\left(\epsilon^{2}_{i}-s^{2}_{1g}\right)\right)=\Pr(D^{(g)}_{i}=1)E\left[\bar{C}_{D}\tilde{X}_{i}\left(\epsilon^{3}_{i}-s^{2}_{1g}\epsilon_{i}\right)|D^{(g)}_{i}=1\right]=0.
\]
Similarly, $E\left(\bar{C}_{D}\tilde{X}_{i}\epsilon_{i}C^{(g)}_{i}\left(\epsilon^{2}_{i}-s^{2}_{0g}\right)\right)=0$.

For the covariance between $\hat{w}$ and $\left(\hat{\gamma}^{\prime}_{s1},\hat{\gamma}^{\prime}_{s0}\right)^{\prime}$, note that $E\left[D^{(g)}_{i}\left(\epsilon^{2}_{i}-s^{2}_{1g}\right)|\tilde{X}_{i}\right]=0$. Since $T_{i}v_{i}$ are functions of $D_{i}$ and hence of $\tilde{X}_{i}$, we have $E\left(T_{i}v_{i}D^{(g)}_{i}\left(\epsilon^{2}_{i}-s^{2}_{1g}\right)\right)=0$. Similarly, $E\left[T_{i}v_{i}C^{(g)}_{i}\left(\epsilon^{2}_{i}-s^{2}_{0g}\right)\right]=0$. For the variance-covariance matrix of $\left(\hat{\gamma}^{\prime}_{s1},\hat{\gamma}^{\prime}_{s0}\right)^{\prime}$, note that all sub-treatment and sub-control groups are mutually exclusive, hence $D^{(g)}_{i}C^{(g^{\prime})}_{i}=0$ and for $g\neq g^{\prime}$, $D^{(g)}_{i}D^{(g^{\prime})}_{i}=0$ and $C^{(g)}_{i}C^{(g^{\prime})}_{i}=0$. As a result, the variance-covariance matrix of $\left(\hat{\gamma}^{\prime}_{s1},\hat{\gamma}^{\prime}_{s0}\right)^{\prime}$ is diagonal. Using (\ref{eq:clt_de2-2}) and (\ref{eq:clt_de2-1-1}) and (\ref{eq:conv_d-1}), the central limit theorem yields $\sqrt{N}(\hat{\delta}^{+}-\delta^{+})\xrightarrow{d}\mathcal{N}(0,\bar{\Sigma}_{\delta^{+}})$, where $\bar{\Sigma}_{\delta^{+}}=\diag\{\bar{\Sigma}_{w},\bar{\Sigma}_{\tau},\bar{\Sigma}_{s}\}$, with $\bar{\Sigma}_{w}$ and $\bar{\Sigma}_{\tau}$ given in Lemma \ref{lem:wtauestimation}, and
\[
\bar{\Sigma}_{s}=\left[\begin{array}{cc}
\diag^{G}_{g=1}\{\frac{2}{p_{T}w_{g}}s^{4}_{1g}\} & 0\\
0 & \diag^{G}_{g=1}\{\frac{2}{(1-p_{T})\varpi_{g}}s^{4}_{0g}\}
\end{array}\right].
\]

To establish consistency of $\hat{\bar{\Sigma}}_{\delta^{+}}$, it suffices to show $\hat{\bar{\Sigma}}_{s}\xrightarrow{p}\bar{\Sigma}_{s}$ as consistency of $\hat{\bar{\Sigma}}_{w}$ and $\hat{\bar{\Sigma}}_{\tau}$ is already proven in Lemma \ref{lem:wtauestimation}. Consistency of $\hat{p}_{T}$ and $\hat{w}_{g}$ is proven in Lemma \ref{lem:wtauestimation}, and consistency of $\hat{\varpi}_{g}$ follows analogously. Since $\hat{\gamma}_{s1}-\gamma_{s1}\xrightarrow{p}0$ and $\hat{\gamma}_{s0}-\gamma_{s0}\xrightarrow{p}0$, the result thus follows from the continuous mapping theorem.

\section{Bias-Corrected Estimator for $\rho_{b}$\label{sec:rho_c}}

While the proposed estimator $\hat{\rho}_{b}$ is consistent for $\rho_{b}$ and asymptotically normal, finite-sample bias arises as demonstrated by Monte Carlo simulation results in Section \ref{sec:MC} in the main text. This section presents a simple bias-correction method for estimating $\rho_{b}$ when sample sizes are small.

The estimator $\hat{\rho}_{b}$ is biased due to the convexity of the exponential function: when $E(\hat{\tau}_{g})=\tau_{g}$, as is the OLS estimator of $\tau_{g}$ in Example 1, $E(\exp(\hat{\tau}_{g}))\neq\exp(E(\hat{\tau}_{g}))=\exp(\tau_{g})$ unless $\Var(\hat{\tau}_{g})=0$. As sample size increases, $\Var(\hat{\tau}_{g})\rightarrow0$, and the difference between $\exp(\tau_{g})$ and $E(\exp(\hat{\tau}_{g}))$ diminishes. The bias is more pronounced with larger $G$, as with more subgroups the size of each group decreases and $\Var(\hat{\tau}_{g})$ increases.

Extending the bias-correction approach of \citet{kennedy_estimation_1981} to accommodate heterogeneous treatment effects yields the bias-corrected estimator:
\begin{equation}
\hat{\rho}_{c}=\sum^{G}_{g=1}\hat{w}_{g}\exp(\hat{\tau}_{g}-0.5\hat{\sigma}^{2}_{\tau,g})-1\label{eq:rho_c-1}
\end{equation}
where $\hat{\sigma}^{2}_{\tau,g}$ is the estimator of the asymptotic variance of $\hat{\tau}_{g}$. Since $\hat{\sigma}^{2}_{\tau,g}\xrightarrow{p}0$ under standard regularity conditions, $\hat{\rho}_{c}-\hat{\rho}_{b}\xrightarrow{p}0$, hence $\hat{\rho}_{c}$ is consistent for $\rho_{b}$.

The bias-corrected estimator utilizes the property that if $x\sim\mathcal{N}(\mu_{x},\sigma^{2}_{x})$, then $E\left(\exp(x)\right)=\exp(\mu_{x}+0.5\sigma^{2}_{x}).$ Hence when $\hat{\tau}_{g}\sim\mathcal{N}(\tau_{g},\sigma^{2}_{\tau,g})$, $E\exp(\hat{\tau}_{g}-0.5\sigma^{2}_{\tau,g})=\exp(\tau_{g})$. As demonstrated in the proof of Lemma \ref{lem:wtauestimation} in Section \ref{sec:examples}, for semi-log regression models as in Example 1, we have $E(\hat{w}_{g})=w_{g}$ as $\hat{w}_{g}$ is an OLS estimator for $w_{g}$ and regularity conditions hold, and $\hat{\tau}_{g}|\hat{w}_{g}\sim\mathcal{N}(\tau_{g},\sigma^{2}_{\tau,g})$ asymptotically because $\hat{\tau}$ and $\hat{w}$ are jointly asymptotically normal and asymptotically uncorrelated. 
\begin{align*}
E\left[\sum^{G}_{g=1}\hat{w}_{g}\exp(\hat{\tau}_{g}-0.5\sigma^{2}_{\tau,g})\right]-1 & =\sum^{G}_{g=1}E\left[\hat{w}_{g}E\left(\exp(\hat{\tau}_{g}-0.5\sigma^{2}_{\tau,g})\left|\hat{w}_{g}\right.\right)\right]-1\\
 & =\sum^{G}_{g=1}E\left[\hat{w}_{g}\exp(\tau_{g})\right]-1=\sum^{G}_{g=1}w_{g}\exp(\tau_{g})-1=\rho_{b}.
\end{align*}
The estimator $\hat{\rho}_{c}$ remains biased for $\rho_{b}$ due to $E\left(\exp(\hat{\sigma}^{2}_{\tau,g})\right)\neq\exp\left(E(\hat{\sigma}^{2}_{\tau,g})\right)$. However, the bias is substantially smaller than that of $\hat{\rho}_{b}$ and decreases quickly with sample size.

\section{Additional Monte Carlo and Empirical Results \label{sec:app_mc}}

\subsection{Skew-normal Error Terms}

In this section, I replicate the Monte Carlo analysis from Table \ref{tab:monteS0truew0} in the main text using skew-normal errors to assess the robustness of $\hat{\rho}_{b}$ and associated inference methods to non-normal errors.

The data generating process is exactly the same as that for Table \ref{tab:monteS0truew0}, except that the error terms $\epsilon_{i}$ now follow an i.i.d. skew-normal distribution, with location parameter $0$, shape parameter $-5$, and scale parameter $1/\sqrt{1-50/(26\pi)}$. As a result, $\epsilon_{i}$ has variance $1$, skewness $-0.851$ and excess kurtosis $0.705$. All other settings remain the same as the simulations using normal errors in the main text. Table \ref{tab:monteS1truew0} displays results, comparable to Table \ref{tab:monteS0truew0}. The results exhibit similar patterns to those observed with normal errors, further supporting the robustness of the proposed estimation and inference methods in this study.
\begin{table}[hp]
\small
\caption{Monte Carlo Results for Estimation and Inference on $\rho_b$: Skew-normal Errors}
\label{tab:monteS1truew0}
\begin{center}
\hspace*{-10mm}
\begin{threeparttable}
\begin{tabular}{cccccccccccccc}
\hline \hline  & \multicolumn{6}{c}{Small Heterogeneity} & & \multicolumn{6}{c}{ Large Heterogeneity} \\ \cline{2-7} \cline{9-14}
& \multicolumn{3}{c}{Estimators} & & \multicolumn{2}{c}{ERR} & & \multicolumn{3}{c}{Estimators} & & \multicolumn{2}{c}{ERR} \\ \cline{2-4} \cline{6-7} \cline{9-11} \cline{13-14}
$ N $ & $ \hat{\bar{\tau}} $ & $ \hat{\rho}_a $  & $ \hat{\rho}_b $ & & $ z_{\tau} $ & $ z_{\rho} $ & &  $ \hat{\bar{\tau}} $ & $ \hat{\rho}_a $  & $ \hat{\rho}_b $ & & $ z_{\tau} $ & $ z_{\rho} $ \\ \hline
 \multicolumn{14}{c}{\textit{True Values}} \\
      & $ -0.201 $ & $ -0.200 $  & $ 0 $ & & $ 5 $ & $ 5 $ & & $ -3.349 $ & $ -3.294 $  & $ 0 $ & & $ 5 $ & $ 5 $ \\
 \multicolumn{14}{c}{\textit{Estimates}} \\20 & 0.10 & 23.31 & 35.67 &  & 6.17 & 6.78 &  & -3.23 & 19.21 & 35.23 &  & 6.10 & 6.85  \\
 & (61.79) & (104.02) & (113.77) &  &  &  &  & (61.76) & (102.58) & (116.02) &  &  &   \\
50 & -0.17 & 7.20 & 11.49 &  & 5.36 & 5.33 &  & -3.39 & 3.93 & 11.43 &  & 5.40 & 5.53  \\
 & (37.15) & (43.77) & (45.38) &  &  &  &  & (37.46) & (42.95) & (46.02) &  &  &   \\
100 & -0.24 & 3.13 & 5.29 &  & 5.19 & 5.17 &  & -3.42 & -0.07 & 5.22 &  & 5.24 & 5.10  \\
 & (25.65) & (27.44) & (27.99) &  &  &  &  & (25.73) & (26.75) & (28.27) &  &  &   \\
200 & -0.20 & 1.43 & 2.58 &  & 5.16 & 5.09 &  & -3.44 & -1.80 & 2.48 &  & 5.47 & 5.11  \\
 & (17.93) & (18.49) & (18.69) &  &  &  &  & (17.98) & (17.96) & (18.82) &  &  &   \\
500 & -0.20 & 0.43 & 1.01 &  & 4.95 & 4.92 &  & -3.39 & -2.71 & 0.96 &  & 6.04 & 5.08  \\
 & (11.24) & (11.37) & (11.43) &  &  &  &  & (11.31) & (11.07) & (11.55) &  &  &   \\
1000 & -0.19 & 0.13 & 0.52 &  & 5.02 & 5.00 &  & -3.34 & -2.98 & 0.52 &  & 6.83 & 4.92  \\
 & (7.96) & (7.99) & (8.02) &  &  &  &  & (7.93) & (7.71) & (8.03) &  &  &   \\
2000 & -0.14 & 0.02 & 0.32 &  & 5.01 & 4.96 &  & -3.36 & -3.15 & 0.24 &  & 9.21 & 5.05  \\
 & (5.60) & (5.61) & (5.63) &  &  &  &  & (5.64) & (5.47) & (5.70) &  &  &   \\
5000 & -0.21 & -0.15 & 0.09 &  & 4.96 & 4.92 &  & -3.34 & -3.23 & 0.11 &  & 15.66 & 5.07  \\
 & (3.54) & (3.54) & (3.55) &  &  &  &  & (3.57) & (3.46) & (3.60) &  &  &   \\
$10^4$ & -0.21 & -0.18 & 0.04 &  & 5.12 & 5.07 &  & -3.34 & -3.26 & 0.06 &  & 26.36 & 5.05  \\
 & (2.50) & (2.50) & (2.50) &  &  &  &  & (2.52) & (2.44) & (2.54) &  &  &   \\
$10^5$ & -0.20 & -0.20 & 0.00 &  & 5.71 & 4.99 &  & -3.35 & -3.29 & 0.00 &  & 98.70 & 5.01  \\
 & (0.79) & (0.79) & (0.79) &  &  &  &  & (0.80) & (0.77) & (0.80) &  &  &   \\
$10^6$ & -0.20 & -0.20 & 0.00 &  & 12.69 & 5.06 &  & -3.35 & -3.29 & 0.00 &  & 100.00 & 5.09  \\
 & (0.25) & (0.25) & (0.25) &  &  &  &  & (0.25) & (0.24) & (0.25) &  &  &   \\
\hline \hline \end{tabular}
\begin{tablenotes}
\footnotesize \item 1. Entries report Monte Carlo means and standard deviations (in parentheses) for $\hat{\bar{\tau}}$, $\hat{\rho}_a$, $\hat{\rho}_b$, based on $ 10^5$ replications for each sample size $N$. 
Also reported are empirical rejection rates (ERRs, in percent) at the 5\% nominal level for two-sided tests of $H_0: \bar{\tau}=0$ using $z_\tau$, and of $H_0: \rho_b=0$ using $z_{\rho}$. 
 \item 2. The left panel uses $\tau=(\ln(0.92),\ln(0.96),\ln(1.04),\ln(1.08))'$, and the right panel uses $\tau=(\ln(0.68),\ln(0.84),\ln(1.16),\ln(1.32))'$. The error terms $ \epsilon $ follow a skew-normal distribution with location parameter $ 0 $, shape parameter $ -5 $, and scale parameter $ 1/\sqrt{1-50/(26\pi)} $ . 
 \item 3. All values are multiplied by 100. True parameter values are in the first row.
\end{tablenotes}
\end{threeparttable}
\end{center}
\end{table}

\subsection{Monte Carlo Simulation for \citet{atkin_exporting_2017}\label{subsec:MC_atk}}

This section conducts simulations calibrated to the parameter estimates and sample structure from the first empirical application \citet{atkin_exporting_2017}. These calibrated experiments evaluate whether the improvements offered by the proposed methods over conventional approximations persist under empirically-realistic conditions. The exercises confirm the reliability of the proposed approach in practical settings.

Data are generated as follows. Corresponding to the data structure in \citet{atkin_exporting_2017}, I first generate a sample of 220 firms, and divide them into 8 strata with sizes $4,48,22,6,52,12,56,20$. Firms in strata 1 to 4 and those in strata 5 to 8 were assigned to treatment groups with probability $1/2$ and $1/4$ respectively. Then firms are expanded three times, representing 3 rounds. I then generate each outcome variable using Eq. (\ref{eq:atk}) in the main text:
\[
\ln(Y_{igt})=\alpha+\sum^{8}_{g=1}\tau_{g}D^{(g)}_{i}+\gamma\ln(Y_{ig0})+\delta_{g}+\theta_{t}+\epsilon_{igt},
\]
where $\ln Y_{ig0}$ is drawn from $\mathcal{N}(0,1)$, $\epsilon_{igt}$ is drawn from $\mathcal{N}(0,\sigma^{2}_{\epsilon})$, and $\sigma_{\epsilon}$ is the RMSE from the original regression for the outcome. Values of $\tau_{g}$ are set to their estimated values for the outcome. I set $\gamma=1$, $\alpha=0$, $\delta_{g}=0$ and $\theta_{t}=0$ for simplicity. \footnote{This is not exactly the same as in \citet{atkin_exporting_2017}, as the original data set is an unbalanced panel with missing values and the treatment probability is unknown.}

\begin{table}[hp]
\small
\caption{Monte Carlo Simulation Results Calibrated to Atkin et al. (2017)}
\label{tab:mc_atk}
\begin{center}
\hspace*{-10mm}
\begin{threeparttable}
\begin{tabular}{cccccccccc}
\hline \hline & \multicolumn{2}{c}{True Values} & & \multicolumn{3}{c}{Estimators} & & \multicolumn{2}{c}{ERR} \\ \cline{2-3} \cline{5-7} \cline{9-10}
Outcome & $ \bar{\tau} $ & $ \rho_b $ & & $ \hat{\bar{\tau}} $ & $ \hat{\rho}_a $ &  $ \hat{\rho}_b $ & & $ z_{\tau} $ & $ z_{\rho} $ \\ \hline \multicolumn{ 10 }{c}{ \textit{Impact of Exporting on Firm Profits} } \\ 
Direct profits & 22.35 & 26.75 &  & 22.34 & 25.16 & 27.65 &  & 6.08 & 4.94  \\
 &  &  &  & (4.42) & (5.54) & (5.70) &  &  &   \\
Reported profits & 18.80 & 22.27 &  & 18.81 & 20.82 & 23.24 &  & 5.98 & 4.92  \\
 &  &  &  & (4.52) & (5.46) & (5.63) &  &  &   \\
Constructed profits & 16.20 & 19.77 &  & 16.18 & 17.69 & 20.75 &  & 6.79 & 4.85  \\
 &  &  &  & (4.67) & (5.50) & (5.71) &  &  &   \\
Hypothetical profits & 33.31 & 49.24 &  & 33.29 & 39.88 & 51.80 &  & 15.20 & 4.71  \\
 &  &  &  & (7.38) & (10.33) & (11.81) &  &  &   \\
\multicolumn{ 10 }{c}{ \textit{Impact of Exporting on Firm Profits (per owner hour)} } \\ 
Direct profits & 17.93 & 21.82 &  & 17.93 & 19.75 & 22.61 &  & 7.30 & 5.03  \\
 &  &  &  & (4.26) & (5.10) & (5.33) &  &  &   \\
Reported profits & 15.97 & 19.23 &  & 15.95 & 17.40 & 20.07 &  & 6.78 & 5.03  \\
 &  &  &  & (4.34) & (5.10) & (5.35) &  &  &   \\
Constructed profits & 13.88 & 16.54 &  & 13.88 & 15.00 & 17.47 &  & 6.08 & 4.88  \\
 &  &  &  & (4.50) & (5.17) & (5.38) &  &  &   \\
Hypothetical profits & 25.91 & 30.89 &  & 25.88 & 29.75 & 32.49 &  & 5.47 & 4.81  \\
 &  &  &  & (5.82) & (7.56) & (7.92) &  &  &   \\
\multicolumn{ 10 }{c}{ \textit{Impact of Exporting on Components of Profits} } \\ 
Output price & 37.60 & 54.01 &  & 37.58 & 45.98 & 56.53 &  & 12.65 & 4.74  \\
 &  &  &  & (7.02) & (10.27) & (11.37) &  &  &   \\
Output & -22.23 & -17.38 &  & -22.22 & -19.76 & -15.81 &  & 7.90 & 4.62  \\
 &  &  &  & (6.34) & (5.09) & (5.60) &  &  &   \\
Hours worked & 4.16 & 4.70 &  & 4.15 & 4.26 & 4.85 &  & 5.54 & 4.99  \\
 &  &  &  & (2.00) & (2.08) & (2.09) &  &  &   \\
Number of looms & -0.72 & -0.46 &  & -0.72 & -0.69 & -0.26 &  & 5.24 & 5.07  \\
 &  &  &  & (2.25) & (2.24) & (2.28) &  &  &   \\
Warp thread balls & 13.29 & 16.31 &  & 13.28 & 14.29 & 17.10 &  & 7.62 & 4.92  \\
 &  &  &  & (3.90) & (4.46) & (4.71) &  &  &   \\
\hline \hline \end{tabular}
\begin{tablenotes}
\footnotesize \item 1. Results from Monte Carlo simulations using a DGP calibrated to the setting in Atkin et al. (2017).
 \item 2. Monte Carlo means and standard deviations (in parentheses) for $\hat{\bar{\tau}}$, $\hat{\rho}_a$, and $\hat{\rho}_b$ across $ 10^5$ replications. Also reported are empirical rejection rates (ERRs) at the 5\% nominal level for two-sided tests of $H_0: \bar{\tau}=\ln(\rho_0+1)$ using $z_\tau$, and of $H_0: \rho_b=\rho_0$ using $z_{\rho}$, where $\rho_0$ is the true value of $\rho_b$.
 \item 3. All values are multiplied by 100. True parameter values appear in the first two columns.
\end{tablenotes}
\end{threeparttable}
\end{center}
\end{table}

Results are summarized in Table \ref{tab:mc_atk}. Overall, the results confirm good finite-sample performance of $\hat{\rho}_{b}$ and inference based on $z_{\rho}$.

\subsection{Imputation-Based Estimates for \citet{alsan_watersheds_2019}\label{subsec:Imputation-AG}}

This section reports estimates of $\bar{\tau}$, $\rho_{a}$ and $\rho_{b}$ for \citet{alsan_watersheds_2019}, using the imputation estimator of \citet{borusyak_revisiting_2024}. I first estimate the following regression on all untreated observations
\begin{equation}
\ln(Y_{it})=\alpha_{i}+\beta_{t}+X^{\prime}_{it}\gamma+\epsilon_{it},\label{eq:stagdid-2}
\end{equation}
where the specification is otherwise the same as regression (\ref{eq:stagdid}), except that the treatment dummies $D_{it}(c,r)$ are dropped. I then obtain the fitted values of $\ln(Y_{it})$, denoted by $\widehat{\ln(Y_{it})}.$ For each treated observation, I calculate the imputed log-point effect $\hat{\tau}_{it}=\ln(Y_{it})-\widehat{\ln(Y_{it})}$, and the corresponding percentage-point effect $\hat{\rho}_{it}=\exp(\hat{\tau}_{it})-1$. For a given set $\mathcal{A}$ defined as in Section \ref{subsec:Example_2_Staggered}, I then compute $\hat{\bar{\tau}}^{imp}=\frac{1}{N_{\mathcal{A}}}\sum_{(i,t)\in\mathcal{A}}\hat{\tau}_{it}$, $\hat{\rho}^{imp}_{a}=\exp(\hat{\bar{\tau}}^{imp})-1$, and $\hat{\rho}^{imp}_{b}=\frac{1}{N_{\mathcal{A}}}\sum_{(i,t)\in\mathcal{A}}\hat{\rho}_{it}$, where $N_{\mathcal{A}}$ is the size of $\mathcal{A}$. For example, $\mathcal{A}$ is the set of all treated observations.

  \begin{table}[h!] \begin{center} \small \caption{Imputation-Based Estimates for the Alsan and Goldin (2019) Application} \label{tab:AG2019_imputation}  \begin{threeparttable}  \begin{tabular}{cccc} \hline \hline &\hspace{20pt} $\hat{\bar{\tau}}^{imp} \hspace{20pt} $ & $ \hspace{20pt} \hat{\rho}_{a}^{imp} \hspace{20pt} $ & $\hspace{20pt} \hat{\rho}_{b}^{imp} \hspace{20pt} $ \\ \hline 
\multicolumn{4}{c}{\textit{ATT for All Treated Units}} \\
&-20.7&-18.7&-9.9\\
\multicolumn{4}{c}{\textit{ATT by Event Time}} \\
0&-17.0&-15.6&-12.7\\
1&-18.6&-16.9&-10.0\\
2&-23.4&-20.9&-15.9\\
3&-19.4&-17.6&-12.0\\
4&-30.3&-26.1&-20.9\\
5&-19.8&-18.0&-6.2\\
\multicolumn{4}{c}{\textit{ATT by Cohort}} \\
1898&-11.1&-10.5&-2.7\\
1899&-21.6&-19.4&-15.7\\
1901&22.8&25.6&30.2\\
1902&-87.3&-58.2&-55.0\\
1903&-36.7&-30.7&-28.0\\
 \hline \hline \end{tabular} \begin{tablenotes} \footnotesize
\item 1. This table reports imputation-based estimates for the Alsan and Goldin (2019) application. \item 2. Estimates are reported for $\hat{\bar{\tau}}^{imp}$, computed as the mean of $\hat{\tau}_{it}$; $\hat{\rho}_{a}^{imp}=\exp(\hat{\bar{\tau}}^{imp})-1$; and $\hat{\rho}_{b}^{imp}$, computed as the mean of $\exp(\widehat{\tau}_{it})-1$ over the relevant treated post-treatment cells. All values are multiplied by 100. \item 3. Panel 1 presents ATT for all treated units. Panel 2 presents ATT by event time $0,1,\ldots,5$. Panel 3 presents ATT by cohort. Standard errors are not reported.
\end{tablenotes} \end{threeparttable} \end{center} \end{table}

Table \ref{tab:AG2019_imputation} reports these point calculations. As discussed in Remark \ref{rem:imputation} in the main text, in the setting considered here, $\hat{\rho}^{imp}_{b}$ has an upper-bound interpretation in expectation.
\end{document}